\pdfoutput=1

\newif\ifpar
\newif\ifmini
\newif\iflong
\newif\ifnconf
\newif\ifreva
\newif\ifArx

\parfalse
\minifalse
\longtrue
\nconftrue
\revatrue
\Arxtrue

\documentclass[10pt,journal]{IEEEtran}
\usepackage[utf8]{inputenc}

\pdfoutput=1

\usepackage[oldcommands,boxed]{algorithm2e}

\usepackage[protrusion=true,expansion]{microtype} 
\usepackage[usenames,dvipsnames]{color}
\usepackage{amsmath, graphicx, latexsym, amssymb, amsthm, amsfonts}
\usepackage[mathscr]{eucal}
\usepackage{arydshln}
\usepackage{color}
\usepackage{epsfig}
\usepackage{url}
\usepackage{psfrag}
\usepackage{empheq}

\usepackage{pdflscape}

\usepackage{wrapfig}

\usepackage[update,prepend]{epstopdf}
\usepackage{wasysym}

\usepackage[subnum]{cases}

\usepackage{cleveref}
\newcommand{\crefrangeconjunction}{--}

\usepackage{comment}
\usepackage{bm}
\newcommand{\BS}{B}

\newcommand{\calS}{{\cal S}}

\newcommand{\E}{{\mathbb E}}
\newcommand{\A}{{\mathbb A}}

\newcommand{\T}{{\mathbb T}}
\newcommand{\R}{{\mathbb R}}
\newcommand{\iy}{{\infty}}

\newcommand{\eat}[1]{ }

\newtheorem{lemma}{Lemma}

\newtheorem{corr}{Corollary}
\newtheorem{theorem}{Theorem}

\newcommand{\req}[3]{\vspace*{#1pt}{{\begin{equation} #2\label{#3}\end{equation}}}}
\newcommand{\reqs}[2]{\vspace*{#1pt}{{\begin{equation*} #2\end{equation*}}}}

\newcommand{\ral}[2]{\vspace*{#1pt}{\begin{align} #2\end{align}}}
\newcommand{\rals}[2]{\vspace*{#1pt}{\begin{align*} #2\end{align*}}}

\newcommand{\rdo}[1]{$#1$}
\newcommand{\sdo}[1]{{\small{$#1$}}}

\newcommand{\bA}{ {\bf A} }

\newcommand{\bP}{ {\bf P} }

\newcommand{\bS}{ {\bf S} }

\newcommand{\bpi}{ {\mbox{\boldmath $\pi$}} }

\begin{document}

\title{An SMDP Approach to Optimal PHY Configuration in Wireless Networks}

% author names and affiliations
% use a multiple column layout for up to three different
% affiliations
\author{\IEEEauthorblockN{Mark Shifrin,  Daniel S. Menasché\IEEEauthorrefmark{4}, Asaf Cohen, Omer Gurewitz, Dennis Goeckel\IEEEauthorrefmark{1}}\\
\IEEEauthorblockA{Ben Gurion University of the Negev}, %\\
%Beer Sheva, Israel}
\and
	%\IEEEauthorblockN{Daniel S. Menasché}
\IEEEauthorblockA{\IEEEauthorrefmark{4}Federal Univ. of Rio de Janeiro}, %\\ Rio de Janeiro, Brazil}
\and
%\IEEEauthorblockN{Dennis Goeckel}
\IEEEauthorblockA{\IEEEauthorrefmark{1}Univ. of Massachusetts at Amherst} %\\
%Amherst, USA}
}

\maketitle
%\vspace{-20pt}
\begin{abstract}
In this work, we study the optimal configuration of the physical layer in wireless networks by means of Semi-Markov Decision Process (SMDP) modeling. In particular, assume the physical layer is characterized by a set of potential operating points, with each point corresponding to a rate and reliability pair; for example, these pairs might be obtained through a now-standard diversity-vs-multiplexing tradeoff characterization. Given the current network state (e.g., buffer occupancies), a Decision Maker (DM) needs to dynamically decide which operating point to use. The SMDP problem formulation allows us to choose from these pairs an optimal selection, which is expressed by a decision rule as a function of the number of awaiting packets in the source’s finite queue, channel state, size of the packet to be transmitted.
We derive a general solution which covers various model configurations, including packet size distributions and varying channels. For the specific case of exponential transmission time, we analytically prove the optimal policy has a threshold structure. Numerical results validate this finding, as well as depict muti-threshold policies for time varying channels such as the Gilber-Elliott channel. 
\end{abstract}

\section{Introduction}{\label{sec:intro}}

Recent advances in coding and modulation have allowed communication systems to approach the Shannon limit on a number of communication channels; that is, given the channel state (e.g. the signal-to-noise ratio of an additive white Gaussian noise - AWGN - channel), communication near the highest rate theoretically possible while maintaining low error probability is achievable. However, in many communication systems, particularly wireless communication systems, the channel conditions which a given transmission will experience are unknown. For example, consider the case of slow multipath fading without channel state information at the transmitter~\cite{tse_vishwanath}. Because of the uncertainty in the level of multipath fading, it is possible that the rate employed at the transmitter cannot be supported by the channel conditions, hence resulting in packet loss or ``outage''.  The outage capacity~\cite{tse_vishwanath}, which gives the rate for various outage probabilities, captures the tradeoff in such a situation.  If a low rate is employed, it is likely that the channel conditions will be such that transmission at that rate can be supported (low outage); if a higher rate is employed, the probability is higher that the channel conditions will be such that transmission at that rate cannot be supported (high outage). In fact, a wide range of physical layers models can be addressed through such an approach, i.e., having asymmetric characterization of the operating points and their corresponding parameters.  Given this characterization, a crucial question is at which point to operate the physical layer given information available about the current state of the network.

Modern wireless networks are extremely dynamic, with channel parameters and traffic patterns changing frequently. Consequently, the key question addressed here is how can a sender \emph{dynamically decide} on the best physical layer strategy, given the channel and traffic parameters available to it, as well as its own status.   For example, consider a sender required to decide among the aforementioned physical layer strategies:  an approach incurring high packet loss yet a small transmission time, or one possibly having a lower packet loss but a larger transmission time. In this paper, we derive a framework for a Decision Maker (DM) wishing to maximize the system throughput by choosing the appropriate physical layer setting, while taking into account as many system parameters as possible, in this case, delay, packet losses and its current packet backlog status. The DM faces a choice of achieving increased success probability provided additional transmission time, and one would expect that this decision will be made with the queue status in mind, as a full queue causes new arrivals to be rejected, incurring potential throughput loss. Thus, our goal is to rigorously analyze this tension, and identify the optimal strategy.
\newline  
The tradeoff between rate and reliability is a fundamental characteristic of the physical layer, and we are interested in this formulation largely because it captures much more than the simple point-to-point single-antenna communication used for illustration in the first paragraph above.   To provide a systematic method to consider how this characterization might be derived, consider the the now-standard diversity-multiplexing tradeoff approach originally applied to point-to-point multi-antenna systems~\cite{zheng_dym_tradeoff} but now extensively extended beyond that.  In particular, the diversity-multiplexing approach captures the fundamental tradeoff between rate (multiplexing) and reliability (diversity) for a number of interesting physical layer choices, including:  (1) point-to-point multiple-input and/or multiple-output (MIMO) systems \cite{zheng_dym_tradeoff}, where the transmitter can decide whether to send multiple streams (``multiplex'') from the multiple antennas or to send a single stream with redundancy (``diversity''), or a combination of the two; (2) half-duplex relay channels (e.g. \cite{karmakar2011diversity}), where the transmitter can decide to use the relay, or not, and how to allocate time to the transmit and receive functions of the half-duplex relay.  We are particularly interested in this latter example, and we will adopt terminology from a classical problem in relaying~\cite{laneman2004cooperative} to help clarify the competing options in succeeding sections.  However, we hasten to remind the reader that the results apply much more generally to the diversity-multiplexing protocol for any physical layer.
\begin{comment}
Our model is also relevant for sources transmitting packets of variable sizes. In particular, once the size of the leading packet in the buffer is known prior to the  the transmission,  selecting the point of the discussed tradeoff also selects the transmission time, which has a crucial impact on the future buffer occupancy. If the buffer is full and arriving packets are immediately rejected, the potential throughput associated with these packets is lost. Clearly, the decision rule and the tradeoff structure are not straightforward in such cases. Nonetheless, we consider them in our model.
\end{comment}
\subsection{Main contribution.}
\vspace*{-5pt}Our main contributions are as follows.
%\begin{itemize}

\textbf{Problem formulation: } We formulate the problem of optimal dynamic PHY configuration for the transmitter with long-lived packets influx by a SMDP. The model is presented in generality, capturing finite buffer size, variable packet size, variable channel state, general transmission time distribution and a decision space which is associated with possible PHY configuration.

\textbf{Value function derivation: } We derive the equations for the value function of the SMDP in several interesting cases. These equations are obtained in a tractable form, allowing a solution by simple value iteration. 

\textbf{Threshold policy characterization: } When transmission times are exponentially distributed, we prove there exists an \emph{optimal policy with a threshold structure}; that is, the source should make use of the more reliable option if the number of pending packets is lower than a given threshold, and transmit with the higher rate option otherwise. We also show that the value function in this case is concave and increasing.

\textbf{Numerical study: } We explore different scenarios by simulations. In particular, we validate the threshold policy and concavity for the exponential case and observe similarity to this structure in other cases as well. In the case of a variable channel (e.g., the Gilbert-Elliott channel) we observe a \textit{multi-threshold} policy which is described by having a separate threshold for each channel mode.

%\end{itemize}	

%While a preliminary paper \cite{s3paper} gave basic numerical results for a related model, 
%
To the best of our knowledge, this work is the first to \emph{analyze} this problem of general optimal PHY operating point selection by SMDP.
\vspace{-15pt}
\subsection{Related work}
\vspace{-4pt}
%The pioneering works on the relay channel  date back to the seventies~\cite{elgamal}, when  
%T. Cover and E. Gamal posed the problem of determining the capacity region of the full-duplex relay channel.  Since then, there has been progress in determining the capacity region for    
%the degraded case, as well as for several MIMO settings \cite{wang2005capacity} or a class of erasure channels~\cite{khalili}. Cooperative strategies and their performance in relay networks were considered in \cite{1499041}.   Nonetheless, in its full generality  the capacity region of the relay channel is   still unknown.    
The general trade-off between PHY rate and reliability has various important applications apart from the already mentioned basic relay channel. 
For example, the trade-off between multiplexing and diversity was discussed in~\cite{heath2002diversity} in the context of MIMO channels, in~\cite{tse2004diversity} in the context of multiple-access channel with fading, and in~\cite{liang2011cognitive} in the context of cognitive radio sensing techniques.
%A key feature, and a source for complications in such relay channels, is that they encompass both space and time diversity. That is, 
The SMDP framework introduced in this paper  accounts for both space and time diversity.  

Various PHY settings naturally reflect the diversity associated with wireless channel.
%Space diversity can be achieved by transmitting simultaneously through two channels, minimizing the effects of fading in a single slot, while time diversity can be achieved since fading also varies over time, hence different schemes can be used at different times. Works  such as
Berry and Gallager~\cite{gallager} and
Collins and Cruz~\cite{cruz} accounted for time diversity and delay constraints, while works such as Scaglione, Goeckel and Laneman~\cite{goeckel} accounted for space diversity. 

% We start our related work list with the pioneering work of \cite{elgamal} on relay, which assumed full duplex channel. 

% Several interesting scenarios were solved since then, e.g., t %Yet, the problem in its most general form is still open. 
Routing solutions in wireless networks by means of queue stabilizations were first addressed in~\cite{tassiulas1992stability}, and more recently in~\cite{yeh2007throughput,supittayapornpong2015achieving}.
%More late works utilizing this approach can be represented by \cite{yeh2007throughput,supittayapornpong2015achieving}. 
Yeh and Berry in~\cite{yeh2007throughput} considered control policies for cooperative relay selection. In particular,  they used maximum differential backlog (MDB) methods to stabilize queues and to achieve the optimal throughput.
The authors mainly concentrated on the transport level, hence did not capture physical level considerations e.g., fading environment. % and diversity increase.
   %which account for queue dynamics in order to optimize both scheduling and routing.  
% in a network that contains  numerous relay nodes.
Our approach is different, as we combine the control on the packet (i.e. transport) level with PHY considerations, 
accounting for fading factors. In addition, note that we assume control of a \emph{finite} queue, hence inherently stable.
Recently, Urgaonkar and Neely~\cite{urgaonkar2014delay} considered a constrained resource allocation problem in a relay network under stringent
 delay constraints.

\cite{supittayapornpong2015achieving} proposes near-optimal throughput proposing algorithms for finite queues. 
We consider a different approach to model the problem, 
namely a semi-Markov decision process (SMDP),  which allows us to analyze a broader set of both PHY and transport level settings and channels, 
to provide an optimal policy, for which we prove structural properties. % as opposed to those considered in~\cite{urgaonkar2014delay, supittayapornpong2015achieving, altman2007constrained}.

The reliability versus delay trade-off in finite buffered wireless networks, accounting for multiple controllers, was 
considered in~\cite{altman2007constrained}.  
%modeled using an %  addressed by 
 % MDP in~\cite{altman2007constrained}, 
 In a game-theoretic setting,  players are coupled through their rewards, and act according to their local buffer and channel states,  
  without being aware of the states and
actions taken by other players. 

For a basic introduction to SMDPs we refer the reader to~\cite{puterman1994markov}. 
A number of previous works on SMDPs have focused on establishing the existence of optimal policies of  
threshold type under a variety of settings~\cite{koole1998structural,hajek1984optimal,sennott2009stochastic,walrand1988introduction}.  
To the best of our knowledge, none of these works have addressed  threshold properties for PHY operating points.

%All these works can

%A large body of work is devoted to the scheduling perspective, assuming one or several relay nodes cannot receive and transmit simultaneously (half-duplex), and focusing on choosing the appropriate relays, optimal routing, etc. For example, Khojastepour \emph{et al.} \cite{aazhang} considered two \emph{modes} of operation in the context of `cheap' relays. Selection relaying schemes were considered in \cite{laneman2004cooperative} as a possible method to achieve cooperative diversity. 

\vspace{-0.1in}

\begin{figure}
	\center
	\includegraphics[scale=0.29]{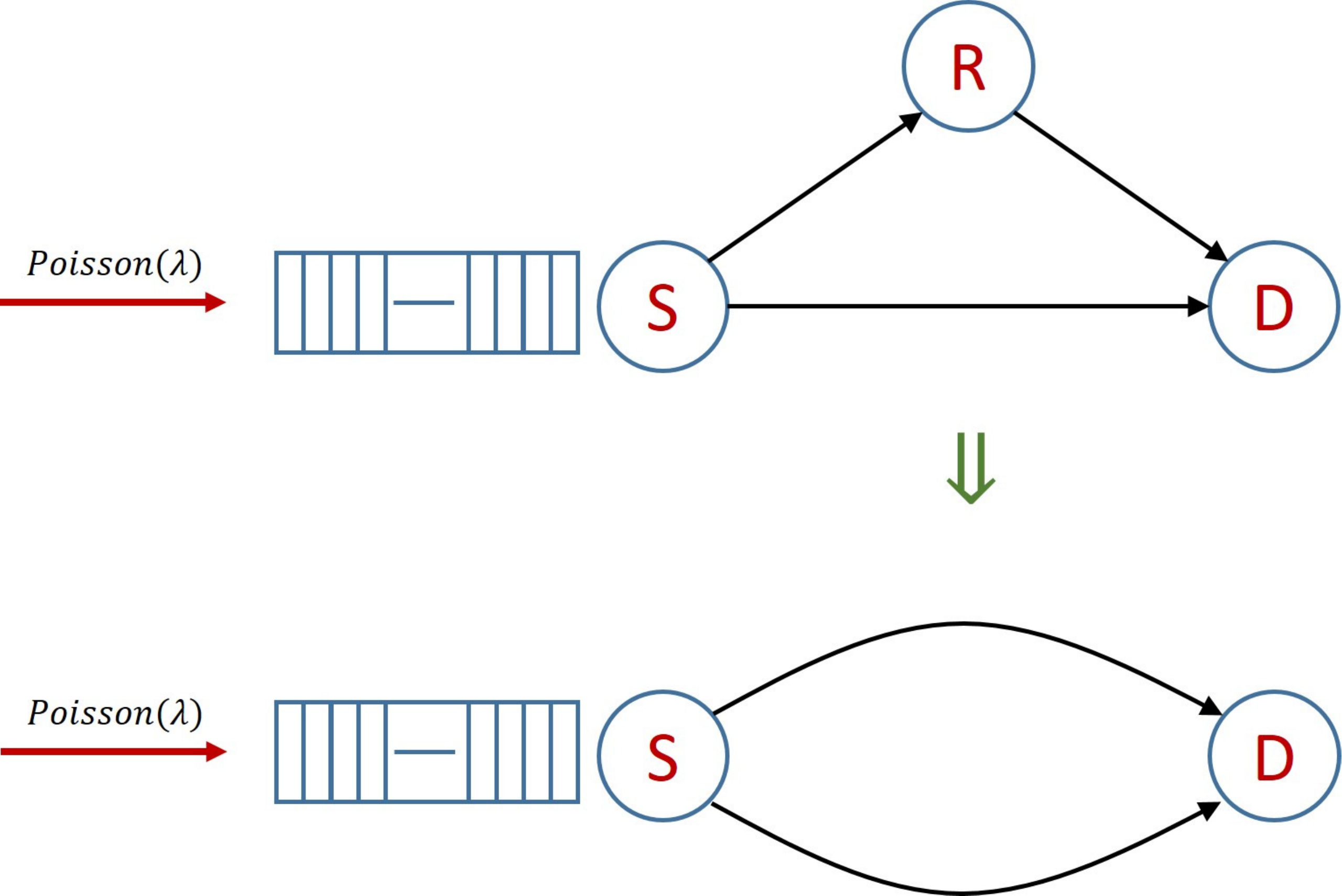}
	\vspace*{-5pt}\caption{Relay channel logical model}\label{fig:32} \vspace*{-15pt}
\end{figure}

\section{System model and metrics}{\label{sec:phy}}
The system is defined according to the following set of assumptions.
\subsubsection{\textbf{Arrival process}} 
We assume that the data at the source
is generated according to a Poisson process with intensity $\lambda$. 
In the case the source is busy with an ongoing transmission, the packets are fed into the \textit{finite sized buffer}.
The size of the buffer is known to the DM and is given by $B$ packets slots.
The packet sizes can vary, yet,
for the model simplicity, it is assumed packets of all sizes occupy space of exactly one slot. Such a model can describe
a setup in which the packets’ descriptors are stored in the
buffer whereas the varied size packets themselves are stored
elsewhere. An extension which considers packets with varying buffer occupancy and specifically ones which occupy a varying
number of \textit{chunks} in the buffer is addressed in~\ref{subsubsec1}.
%
%We consider packet sizes can vary, yet,
%for the model simplicity, it is assumed packets of all sizes occupy space of exactly one slot of a \textit{finite sized buffer}. (For example, consider the case where packet descriptors are stored). The straightforward extension to consider packets occupying in the buffer a varying number of \textit{chunks}  is addressed in~\ref{subsubsec1}. 
Packets arriving at a full buffer are rejected and the retransmission details are taken care of by higher layers.
%Nevertheless, we assume that the packet sizes 
%With slight abuse of terminology, will use the term packet in sequel.
\subsubsection{\textbf{Operating points}} 
The source attempts to communicate the packet at the head of the queue to the destination by a choice among the possible \textit{operating points} of the physical layer. 
For presentation and analytical simplicity, we assume at most two possible operating points at all decision epochs. To exemplify this particular scenario, consider a simple half-duplex relay channel. In this case, we are motivated by the application of the diversity-multiplexing tradeoff to a classic half-duplex relay channel choice: the more reliable and lower relaying rate and outage probability from a prototypical early cooperative diversity protocol, and the less reliable and higher rate corresponding to the direct choice.
While this choice applies cleanly to this simple example, the extension to the situation of more than two possible operating points, as might be appropriate for other PHY layer architectures, is straightforward.
Henceforth, in all scenarios mentioned in this paper, a nomenclature of half-duplex relay channel will be used. We will always denote the more reliable path with a lower rate as path ``a", while the less reliable path with a higher rate  will be denoted as path ``b". Denote the rates $R_a$ and $R_b$, where $R_a<R_b$. The corresponding packet loss probabilities are denoted as $p_a$ and $p_b$, where $p_a<p_b$ at all times.
The reader should note that the limitation to only two operating points rules out the option of abstaining from a transmission. While the latter case opens an additional interesting tradeoff in certain cases (e.g.,~\cite{altman2007constrained,s3paper}), we omit it here for the sake of simplicity of presentation. Yet, the straightforward extension to include the option of not to transmit for the certain time period is technically tractable; the interested reader is referred to the on-line version of the paper for the details.
\subsubsection{\textbf{Transmission times}}\label{subsub:Tt}
The rates associated with the two operating points yield transmission time probability density functions (pdf) $g^{a}(t)$ and $g^{b}(t)$, which are readily calculated from $R_a$ and $R_b$ and, possibly, from other system parameters relevant at time $t$, as it is further specified in~\ref{subsub:Csd} and~\ref{subsub:Psi}. Note that in the case of a half duplex relay model $g^{a}(t)$, refers to the entire path associated with choice ``a", even if the latter is subdivided into two separate paths. Therefore, the example of a half duplex relay obtains a simplified configuration, as Figure~\ref{fig:32} demonstrates. 
See that the upper illustration particularly fits the relay channel problem, while the lower one can also refer to the general tension between two PHY settings associated with two different propagation paths; the upper path corresponds to the choice of reliability (diversity) while the lower path corresponds to the choice of rate (multiplexing). 
\subsubsection{\textbf{Packet size impact}}\label{subsub:Psi}

The transmission time is allowed to depend on the actual packet size. We consider \textit{two modes} to capture such a dependency. In the first mode, the packet size cannot be timely sampled by the DM, hence is unknown prior to the transmission. Then, transmission times are random according to what is specified in~\ref{subsub:Tt}. % as it is specified in sequel. 
\begin{comment}Note that this can model the scenario where application produces chunks and starts the transmission before the last chunk of the packet has been produced. 
\end{comment}
In the second mode, the size can be sampled prior to the transmission. Then, the size of the packet to be transmitted is a part of the \textit{state} and has impact on the decision. Denote the size of a packet in this case by $k$.
%Hence, in the case the packet size is known, it will determine the transmission time. Otherwise, the transmission time is given by a corresponding probability density function (pdf), denoted by \rdo{g^u(t)}.
Then, \textit{packet size transition probabilities}, denoted by $q(k'|k)$, stand for the probability of 
having packet of size $k'$ at the head of the queue after transmitting packet of size $k$. % by acting $u$, where $u\in\{a,b\}$.
%\iflong And, as in the case of the channel transition probabilities, these probabilities can be different for the empty buffer state to allow for the possibility of a different size distribution for the first packet to arrive into the empty buffer. \fi  
%As we further specify in~\ref{subsub:ti}, the packet size can impact the transmission time distribution.
To this end, we assume a finite set of possible packet sizes such that \rdo{\sum_{k'}q(k'|k)=1}. %for each possible 
action $u$. 
The packet sizes can be coupled with transmission time distribution, hence, the corresponding pdf are given. We denote them by $g^{a}_k(t)$ and $g^{b}_k(t)$, and in the second mode they can be naturally assumed to be deterministic. In the detailed example presented in~\ref{sec:det} and in numerical study of this paper, we assumed the same packet loss probability for all packet sizes. However, once the size of the packet to be transmitted is a part of the state, the coupling of the packet size with the packet loss (e.g. by accounting for the BER) is straightforward and is transparently incorporated in our model. %without additional technical extensions. Hence, the model which 
%Hence, the transmission time distribution can depend both on the size of a packet to be transmitted and on the other parameters which impact the system state.  
Note that $q$ can capture complex packet arrival patterns. %, rather than mere Poisson process. 
For example, a sequence of big packets which is likely to be followed by sequence of small packets and vice versa can be modeled provided the appropriate values for $q$ are selected. 
%Also note that in the case of continuous intervals of operating points, the model can be addressed by known techniques for continuous action space, e.g., discretization.

\subsubsection{\textbf{Channel states and dynamics}} \label{subsub:Csd}
We consider a finite set of possible channel states, each state corresponding to a couple $h=(h_a,h_b)$. 
%The channel state corresponds to one of the states. 
We assume that channel state dynamics can be modeled by a Markov process; that is, the \textit{current channel statistics} are independent of the past given the last state of the channel. These dynamics reflect well-known channel models (e.g., i.i.d. fading, Gilbert-Elliott model. We provide a detailed example in~\ref{sec:uniform}). 
%Denote the channel state by $h$, and t
The channel transition probability from state $h$ to state $h'$ is denoted by $p(h'|h)$. %, where the superscript $u$, $u\in\{a,b\}$ denotes the transmission mode. 
Clearly, \rdo{\sum_{h'}p(h'|h)=1}.  
Next, consider a sequence of fading values observed across source packets, each value corresponding to a channel (or, separately, to each path) state.  
We assume that the statistics are known to the DM and are constant for a period which is significantly longer than the longest possible packet transmission time.%, hence independent of transmission mode. 
The DM can obtain the packet loss probabilities associated with each such fading value for each potential PHY configuration.
%Note that we allow generality of these probabilities by assuming they can be different for different transmission modes.  
%\iflong Also, for the empty buffer state, we allow these probabilities to be distinct from those at other buffer states, hence capturing the time it takes until the first packet arrival to the queue. \fi 
While these assumptions are approximations, they conform to the well-known realistic slow fading model, or quasi-static channels.
Similarly to the impact of a packet size, the channel states can be coupled with transmission time distribution, hence, the corresponding pdf are given. We denote them by $g^{a}_h(t)$ and $g^{b}_h(t)$, and, in the general case, by $g^{a}_{h,k}(t)$ and $g^{b}_{h,k}(t)$. Clearly,  different channel states correspond to different packet loss probabilities. %We assume that DM can knows the corresponding packet loss probability at all times. 
In the Gilbert-Elliott model, for example, these probabilities can be calculated from the BER.

%relaying (``r") choice, with rate $R_r$
%with rate $R_d$ and outage $p_d$. 
%{\bf DG:  Not sure I quite understand "For convenience, we will refer to the relay transmission time as an atomic period." in the new context. MSh: let's use the clarification in blue instead? }  
%We consider packets of random length, hence, 

%\newline 

% To finalize, the tension is clearly understood: a higher transmission rate choice that uses less resources (time) but at a higher packet loss probability, or a lower transmission rate choice that exploits more resources to obtain a lower packet loss probability.

The tension between the operating points is summarized as follows: 
a higher transmission rate choice  uses less resources (time) but at a higher packet loss probability, whereas
 a lower transmission rate choice  exploits more resources to obtain a lower packet loss probability.

\subsubsection{\textbf{General performance criterion}}
%We now define the performance criterion. %, which will suit our SMDP formulation as defined in the next section.
%\iflong
%The performance criterion is the expected total discounted infinite horizon reward which is given by the following well-known definition %\iflong
%
%\rals{-5}{
%J=&\E\int_0^\infty e^{-\gamma t} r(t)dt=\E\sum_{m=0}^\infty\int_{\sigma_{m}}^{\sigma_{m+1}} e^{-\gamma t} r_mdA(t)}	%\fi
%
%where $\gamma$ is a discount factor and $r(t)$ is a reward accumulated at time $t$. The equality is reasoned by the fact that in our model rewards are only awarded at discrete time points (e.g. transmission opportunities ends) denoted by $\sigma_m$, $m\in\{0,1,\cdots\}$ that is, where the differential of counting process $A$, counting the transmission attempts and formally defined by 
%\[A(t)=\sup\{l\geq0:\sum_{i=0}^l\tau_i<t\},\] where $\tau_i=\sigma_i-\sigma_{i-1}$, is positive. However, in general, the reward function $r(t)$ allows to capture continuous cost, e.g. energy cost which can be proportional to the buffer occupancy at time $t$. Note that $r_m$ may be equivalent to a size of the successfully decoded packet transmitted in transmission opportunity $m$, or to a constant reward $r_m=C$, regardless of its size; it may incorporate additional cost factors predefined by DM. Without loss of generality, we assume that lost or dropped packets yield no reward. 
%\else
%The performance criterion is the expected total discounted infinite horizon reward which is given by the following well-known definition %\iflong
Consider an SMDP with discounted cost functional, characterized  by the tuple \sdo{\{\bS,\bA,\bP,r, \gamma\}}, where the components are as follows. %stand for the state space, the action space, the state transition probabilities, the reward function and appropriately selected constant standing for a discount factor. 
The state space $\bS$ expresses the set of all possible combinations of buffer state, channel state, and leading packet size. The action space $\bA$ is a set of actions which fits exactly one action to one operating point. The reward $r$ is positive when a transmission was successful. %The positive reward is a symbolic representation of the instantaneous throughput and, finally, 
The discount factor $\gamma$ is an appropriately selected positive constant.
%\newline

A decision must be made before every transmission.  %Consider the embedded process wherein transitions occur at points where 1) a decision must be made and/or 2)  reward is potentially accumulated.    In the embedded process, transitions occur after every service completion and after arrivals to an empty system.   
Let $\{ \sigma_m\}$, $m=0,1,\ldots$,  be a time series where $\sigma_m$ is the instant at which the $m$-th reward is added. %state transition of the embedded process occurs.  We assume $\sigma_0=0$. 
%Positive rewards are accumulated at the end of successful transmission completions.  
Let $r_m$ be a random variable characterizing the $m$th reward. %associated to the $m$-th state transition of the embedded process, which is a function of the policy $\pi$.  
Then, $r_m > 0$  when $\sigma_m$ corresponds to a completion of a successful transmission; otherwise, $r_m = 0$. Since the average discounted infinite cost depends on the initial system state, 
it is given by $J^\pi(s_0)=\E \sum_{m=0}^\infty r_m e^{-\gamma \sigma_m}$.
% J^\pi(s_0)=\E\sum_{m=0}^\infty e^{-\gamma\sigma_{m}} r_m
% Denote the ending time of $m$-th transmission attempt as $\sigma_m$, $m\in\{0,1,\cdots\}$. Clearly, positive rewards can be only added up at $\sigma_m$. The rewards, $r_m$ are r.v's, with expected value calculated in correspondence with some given policy $\pi$. Namely, in the case the transmission at time $\sigma_m$, according to selected action is successful, the reward was positive. In all other cases, $r_m$ is equal to zero.
% Define $J^\pi(s_0)=\E\sum_{m=0}^\infty e^{-\gamma\sigma_{m}} r_m$.
% In order to apply the dynamic programming principle, we view states right after a transmission starts. Define the augmented time series $\eta_m$, $m\in\{0,1,\cdots\}$, such that in the case a transmission is initiated at a non-empty system, the average discounted reward given by $\E r_me^{-\gamma\eta_{m}}$ is awarded. Hence, in the case a transmission starts right after arrival to the empty buffer, $r_m=0$, at $\eta_m$.
%\newline 
Write the discounted cost as the sum of the reward obtained at the initial state and average discounted residual throughput profit associated with future rewards. 
%
%\small
\begin{comment}
\req{-5}
{J^\pi(s_0)=\E\sum_{m=0}^\infty e^{-\gamma\sigma_{m}} r_m =\E r_0^\pi(s_0)+\E\sum_{m=1}^\infty e^{-\gamma\sigma_{m}} r_m}
{eq:disc}
\normalsize
\end{comment}
\begin{align}
& J^\pi(s_0) =\E r^\pi(s_0)e^{-\gamma\sigma_{0}}+\E\sum_{m=1}^\infty e^{-\gamma\sigma_{m}} r_m  \nonumber \\
& =\E r^\pi(s_0)+J^{\pi}_1(s_0)  \label{eq:disc}
\end{align}% \end{align}

Thus, our goal is the following:

Find the \textit{dynamic physical layer setting selection policy $\pi$}, which takes as input the current occupancy of the buffer and information (if any) on the current channel state and the packet size at the source's queue head, and provide as an output a decision on which transmission path to employ, such that $J^\pi(s_0)$ is maximized.  

Note that the analysis that follows can be easily extended to account for an average cost criteria, defined by 
$J^A=\lim_{N\to\infty}\frac{1}{N}\E\sum_{m=1}^N  r_m$ in the sense that both criteria possess  \textit{similar optimal policies}. 
The connection between discounted and average infinite reward criteria is understood via Blackwell optimality as it is demonstrated in e.g.,~\cite{bertsekas1995dynamic}.
%In particular, under mild conditions on $\gamma$, they  referred to as Blackwell optimal~\cite{bertsekas1995dynamic}.

 %Hence, average reward for action $u$ is given by $C(1-p_u)$. 
%Define average reward over infinite horizon
%In this paper, we only allow lump rewards at ending times of successful transmission attempts, denoted by $\sigma_m$, $m\in\{0,1,\cdots\}$.  We assumed a constant reward $r_m=C$, regardless of a packet size, however our model covers any lump reward function.

%\rals{-9}{
%J=&\E\int_0^\infty e^{-\gamma t} r(t)dt=\E\sum_{m=0}^\infty e^{-\gamma\sigma_m} r_m}	%\fi
%J^A=&\lim_{N\to\infty}\frac{1}{N}\E\sum_{m=1}^N  r_m} %=\E\sum_{m=0}^\infty e^{-\gamma\sigma_m} r_m}
%\iflong\else\normalsize\fi
%where $\gamma$ is a discount factor and $r(t)$ is a reward accumulated at time $t$.  %it may incorporate additional cost factors predefined by DM. Without loss of generality, we assume that lost or dropped packets yield no reward. 
%\fi%\iflong The second inequality merely solves the integral.\fi %Note that the expected value is over the interval boundaries as. 

Hence, we formulate the described problem by a discounted SMDP, which is formally explained in the next section.

\begin{table}[h!]
\center
\begin{tabular}{l|l}
\hline
variable & description \\
\hline
% $\bS$ & state space of SMDP  \\
% $\bA$ & action space \\
% $\calU $ & policy space  \\
                $\lambda$ & arrival rate \\
$\BS$ & buffer capacity (including the packet being transmitted) \\\hdashline
% $N$ & random variable characterizing number  of packets in  & $N \in \{0,1,\ldots,B-1 \}$  \\
% & queue, i.e.,  waiting queue + under transmission & \\ 
%             & (as states are tracked  right after  a service & \\ 
%          &   completes, $N \leq B-1$) & \\   
$s$ & SMDP state (buffer state,channel state,leading packet size) \\
%& Note: as state transitions occur after departures, the buffer \\
% & state after a transition  ranges between 0 and $B-1$ \\      
   $V(s)$ &   value function at state $s$, $V(s)=\max_{u\in\{a,b\}}\{V^{u}(s)\}$ \\
    $V^{(u)}(s)$ &   cost associated with decision $u$ at state $s$, $u \in \{ a,b\}$   \\
  %   $V^{(b)}(s)$ &  cost associated with decision ``b" at state $s$     \\   \hdashline
    %$c_r$ & instantaneous relay throughput (reward) & $(1-p_r)C\beta_r$ \\
     %   $c_d$ & instantaneous direct throughput (reward) & $(1-p_d)C\beta_d$ \\
        $\beta_{u,s}$ &  discount associated with action $u$ \\
        $k$ &  number of frames  per packet \\
                     $\tau_{u,s}$ & mean time to transmit via channel $u$ at state $s$ \\
       %  $\tau_{b,s}$ & mean time to transmit via channel ``b" at state $s$ \\
                $\mu_{u,s}$ & transmission rate of channel $u$,          $\tau_{u,s} = 1/\mu_{u,s}$ \\ \hdashline
       % $\mu_{b,s}$ & transmission rate of channel ``b",           $\tau_{b,s} = % 1/\mu_{b,s}$ \\ \hdashline
        $p_u$ & packet loss probability, $u \in \{a,b\}$ \\ 
  %      $p_b$ & packet loss probability associated with channel ``b"\\   \hdashline
%        $\gamma$ & discount factor \\ 
%        $\pi$ & policy (subscript) \\ \hdashline
        $g^{\pi}(t)$ & pdf of transmission time then using policy $\pi$ \\ 
        $P^u(j|i,t)$ & probability of $j-i$ arrivals after $t$ time units, given a \\
        & buffer initially filled with $i$ packets and action  $u$  is taken  \\
        $p(h'|h)$ & channel transition probabilities \\
        $q(k'|k)$ & packet size transition probabilities\\%, given  action  $u$ is taken \\
       \hline
\end{tabular}
\vspace{-2pt}
\caption{SMDP notation.} \vspace{-0.25in} %\label{tab:general}
\end{table}

\vspace*{-5pt}\section{SMDP-based formulation and solution} \label{sec:mdpmodel}
%In this section, we define the general form of the SMDP and give specific examples.
%We consider an SMDP with discounted cost functional, characterized  by the tuple \sdo{\{\bS,\bA,\bP,r, \gamma\}},
%where the components stand for the state space, the action space, the transition probabilities, the reward function and  the discount factor.
%Next, we introduce the SMDP formulation for the optimal PHY configuration problem.  
We start with definition of the state-space. A state ${s\in\bS}$ is expressed by the triplet
$(n,h,k)\in\R^3$ where $n$, $h$ and $k$  stand, respectively, for 1) the number of packets present in the system (including the one being transmitted), 2) the medium (channels) state and 3) the size of the packet which is to be transmitted in the upcoming transmission.
%State transitions occur after transmission completions.
%\newline

Transmissions can take one of two possible paths (``a" or ``b"). Correspondingly, the action space is given by $\bA=\{a,b\}$, standing for transmission modes ``a" and ``b". We assume that a  transmission decision with corresponding parameters is performed at every decision epoch.

%\newline

State transitions occur after 1) transmission completions and 2)  arrivals to an empty system.
Right after a state transition due to a transmission completion occurs, a decision which sets the  transmission mode of the next  packet to be transmitted is made. If the system is non empty, a new transmission  is immediately started. Otherwise, it  starts whenever an arrival occurs.

%\newline

The probability of having $j$ packets in the buffer after a transmission of a packet of size $k$, taking $t$ time units,
when the buffer initially contained $i$ packets, is denoted by \rdo{\varrho(j|i,t)} and
 is governed by a Poisson distribution with mean ${\lambda t}$.  If $j<B+1$ and $i>0$ then $j=i-k+m$,
 where $m$ is the number of arrivals during an interval of length $t$.
%(\textbf{TODO} - or 2nd+3rd can be seen as additional single dimension)
%However, they can depend on the action attempted in the same slot.
%The action space is defined by $\bA=\{a,b\}$, standing for transmission mode ``a" and transmission mode ``b",
%and abstaining from transmission. %In the case the buffer was empty or action $0$ was selected,
%no transmission is initiated, while the next decision will be performed right after the next packet is ready.
%Otherwise, the actions are taken right after the accomplished transmissions.
%The number of arrivals between two consecutive departures is given by Poisson distribution.
%Note that if the source is not empty, a packet is transmitted immediately after the preceding transmission end.
Let $r$ be the instantaneous gain at the end of a successful transmission.  It is given by $k$, in  case the reward is set according to  packet sizes, or constant equal $1$ if successfully transmitted packets of all sizes have the same value.
%\newline

Next, we expand the~\eqref{eq:disc}. Note that the second term of~\eqref{eq:disc} is given by ${\E_{(\sigma_1,s_1)}[ e^{-\gamma\sigma_1}V(s_1)]}$, where $s_1$ is the state following $s_0$.
The superscript $\pi$ stands for the policy under consideration.
Note that the value function $V$ is obtained by maximizing $J$  over all feasible  policies,
 and  is given by $V(s_0)=\max_\pi J^\pi(s_0)$, for all $s_0\in\bS$.
%For cost, which treats the general state formulation can be expressed as follows. \sdo{s_0=\{i,h,k\}} and \sdo{s_1=\{j,h',k'\}} write
%In order to simplify the analysis, we assume that the reward for the transmission is accumulated at the \textit{beginning} of the transmission. As long as the time distribution of the upcoming transition is known, we weigh that reward by the corresponding average discount at the end of the transmission.
% We now expand for all possible transitions.

Denote the transition probability from state $s_0=(i,h,k)$ to state $s_1=(j,h',k')$ by $P(s_1|s_0,\pi(s_0),t)$. $P(s_1|s_0,\pi(s_0),t)$ depends on the  arrival process, channel dynamics and leading packet size, and is given by
\begin{align}
& P(s_1|s_0,\pi(s_0),t)=P((j,h',k')|(h,k,i),\pi(s_0),t) = \nonumber\\
& = q(k'|k)p(h'|h)\varrho(j|i,t)
\end{align}
where $\sum_{j,k',h'}q(k'|k)p(h'|h)\varrho(j|i,t)=1$.
The Bellman equation for the initial state $s_0$, action $\pi(s_0)$ and  next state $s_1$ is given by~\eqref{eq:disc}, setting ${J}^{\pi}_1(s_0) $ as follows
\begin{align}
&{J}^{\pi}_1(s_0) = \E e^{-\gamma t}\sum_{s_1}V(s_1)P(s_1|s_0,\pi(s_0)=u_0,t)\label{eq:be0}\\
&=\sum_{s_1}V(s_1) \int_0^\iy e^{-\gamma t}P(s_1|s_0,u_0)g^{u_0}(t)dt\label{eq:be1}\\
&=\sum_{h'}\sum_{k'}\sum_{j=i-k}^{B-k+1}V(s_1)\int_0^\iy e^{-\gamma t}P(s_1|s_0,u_0)g^{u_0}(t)dt.\label{eq:be2}
\end{align}
\begin{comment}
\iflong\else\footnotesize\fi
\ral{-5}{
& J^\pi(s_0)= \E r^\pi(s_0)+\Big[\sum_{s_1=(h',k',j)}V(s_1)\cdot\nonumber\\
&\int_0^\iy e^{-\gamma t}P(s_1|s_0,\pi(s_0))g^{\pi(s_0)}(t)dt\Big]=\label{eq:be1}\\
&\E r^\pi(s_0)+ \Big[\sum_{h'}\sum_{k'}\sum_{j=i-k}^{B-k+1}V(s_1)\cdot\nonumber\\
&\int_0^\iy e^{-\gamma t}P(s_1|s_0,\pi(s_0))g^{\pi(s_0)}(t)dt\Big].\label{eq:be2}}
\iflong\else\normalsize\fi
%\end{equation}}}
\end{comment}
%
%
% $\E r^\pi(s_0)$ is the expected average reward at the initial state $s_0$.
The first (i.e the outer) summation in~\eqref{eq:be2} is over all possible next channel states.
 It is degenerated if the channel state is fixed. The second summation is over all possible packet sizes to be transmitted at state $s_1$. If the packet size is unknown at  decision time, or all packet sizes are equal, this sum degenerates. The third summation is over the number of arrivals to the queue during a transmission. The integration accounts for the transmission time. Note that the transmission time pdf $g^{\pi(s_0)}$ may depend on the action taken in state $s_0$.
 The expected instantaneous reward at state $s_0$ accounts for the average discount at the end of the transmission:
\vspace*{-5pt}\begin{equation}
\E r^\pi(s_0)=\int r^{\pi(s_0)}e^{-\gamma t}g^{\pi(s_0)}(t)dt,\label{eq:8}
\end{equation}
where % $r^0=0$,
 % while
   $\pi(s_0) \in \{a,b\}$, $r^{a}=k(1-p_a)$ and $r^{b}=k(1-p_b)$.

%\end{remark}

%Next, we consider special cases wherein~\eqref{eq:be2}  admits a closed form expression.  As a consequence, the Bellman equations are tractable.   In these cases, a straightforward application of value iteration can be used to find the optimal policy~\cite{puterman1994markov}, and the solution is guaranteed to be unique.

%\vspace{-0.1in}

% TWO PARAGRAPHS BELOW WERE PUT INTO COMMENT
%
% Note also that~\eqref{eq:be2} is a weighted sum of $V(s_1)$, for  $s_1 \in \bS$, where the weights are given by the integral of   $P(s_1|s_0, \pi(s_0),t)$ times a probability density function and an exponential discount.
The tractability of the Bellman equation form depends on resolvability of the integration in~\Crefrange{eq:be0}{eq:be2}. In the case a closed form can be obtained, a convenient utilization of value iteration (see, e.g., ~\cite{puterman1994markov}) procedure is possible. In particular,~\Crefrange{eq:be0}{eq:be2} substituted in~\eqref{eq:disc} result in recursive equations. Then, the repetitious application  of~\eqref{eq:disc}, guarantees an arbitrary closed convergence of $V$ to the fixed point, which stands for the unique solution. % can be an arbitrary closed.

% The Bellman equations admit closed form when the integration in~\Crefrange{eq:be0}{eq:be2} yields   a simple expression.  In this case,    the problem is amenable to solution through value iteration,

%\ifArx
\subsubsection{Chunk granularity}\label{subsubsec1}
In order to consider packets, which can contain a variable number of chunks, several technical adjustments should be made. In particular, the Poisson \textit{chunk arrival process} would feed the buffer room, which should be measured in \textit{chunk units}. The packet size can be indicated within the last or the first chunk in a packet (i.e, in correspondence with the two modes of packet size impact on the decision making).
%In order to incorporate a general definition of having packets larger than one chunk,
The state definition should include the number of chunks in the buffer, hence,~\Crefrange{eq:be0}{eq:be2} should be accordingly adjusted to reflect the change in the queue size after each transmission. In addition, the calculation of~\eqref{eq:8} should account for all additional states where the number of chunks present in the buffer is less than one packet.
\iflong
\vspace*{-15pt}\subsection{Exponentially distributed transmission times}\label{sub:exp}
%Next, we consider exponentially distributed transmission times.  
Let $1/\mu_u$  be the expected transmission time,  $u \in \{a,b\}$. %as estimated by the controller.   
%Denote exponentially distributed transmission times ${1}/{\mu_u}$, \rdo{u\in\{r,d\}}, with no packet size knowledge prior to transmission. The relay path ``r" stands for the mode ``a", while the direct path ``d" stands for the mode "b". \iflong 

\subsubsection{General SMDP formulation}
%
% In this section we assume that   channel ``a'' corresponds to  relayed transmissions, whereas channel ``b'' corresponds to a direct link between the source and the destination.  
% In addition, we assume that the controller does not condition its decisions on the size of the  packet  at the head-of-line,  prior to transmission (in Section~\ref{sec:det} we remove this assumption). 
%
%
% Note that during a  transmission through channel ``a" the time taken for the first transmission, between the source and the relay,   is assumed to be known prior to the beginning of the second transmission, between the  relay and the destination, and  the duration of  the second transmission is assumed to be deterministic given the duration of the first. Under these assumptions,   transmission times are exponentially distributed, irrespectively of whether channels  ``a" or ``b" are used. 
%
% \else We assume that duration of the second interval of a relay transmission is deterministic. Hence, it is easy to see that the total relay transmission time is exponentially distributed. \fi 
%
The state space is one-dimensional. % Let $s=\{n\}$, i.e., t 
Henceforth, we assume $B \ge 2$.  
The system state is characterized by the  number of packets at the source, including the packet being transmitted,
 and is denoted by $s$, $s \in \{0, 1, \ldots, B-1\}$. We assume all packets are equally valued with  associated reward $1$.
%In order to be able to make a comparison we now always view the system at points there the packet is served.
%The general solution of an SMDP is given by equation (11.3.6) in~\cite{puterman1994markov}, and it is based on seeing the system states between subsequent service events. In particular, one writes
Let  $g^u(t)=\mu_u e^{-\mu_u t}$, where superscript $u$ stands for the action taken, $u \in \{ a, b\}$.  
Then, it follows from~\eqref{eq:be1} that
%\begin{equation}
%V^\pi(s_0)=r(s_0)+\sum_i \left( \int_0^\iy e^{-\gamma t}\varrho(s_1=i|s_0,t)g_i(t)dt \right) V^\pi(s_1^i)
%\label{be3}
%\end{equation}
%Next, we specialize the result above to our setting.  Note that for $a \in \{ d,r \}$,
%\begin{eqnarray}
%g_i(\tau_a) &=& \mu_a e^{-\mu \tau_a} \\
%\varrho(s_1 = i|s_0,t) &=&  \frac{e^{-t \lambda} (\lambda t)^{i-j+1}}{(i-j+1)!} \\
%r(s_0) &=& C(1-p_a)  E[e^{-\gamma \tau_a}]
%\end{eqnarray}
%Then, for $0 < j < Q$,
% \rals{-5}{
\begin{align}
 %J^{u}(j)=  (1-p_u) E[e^{-\gamma \tau_u}]+ \nonumber
 & J^{u}_1(j)= V(B-1)\int_{0}^{\iy}e^{-\gamma\tau_u}\varrho(B|j,\tau_u)\mu_u e^{-\mu_u \tau_u} d\tau_u +  \nonumber\\
& +\sum_{i=j-1}^{B-1} V(i)  \left( \int_{0}^{\iy} e^{-\gamma\tau_u}\frac{e^{-\lambda\tau_u}(\lambda \tau_u)^{i-j+1}}{(i-j+1)!}\mu_u e^{-\mu_u \tau_u} d\tau_u \right)  \nonumber
\end{align}
where \rdo{\varrho(B|j,\tau_u)=\left(1-\sum_{i=j-1}^{B-1}\frac{e^{-\lambda\tau_u}(\lambda \tau_u)^{i-j+1}}{(i-j+1)!}\right)}  denotes the probability that during a transmission at least one packet is discarded due to buffer overflow, given that there are $j$ packets in the system prior to the beginning of the transmission.
%   during aby the transmission end the buffer is full is given by subtraction from $1$ of the sum of all probabilities which account for the cases when the number of arrivals did not fill the buffer to the maximum capacity.
%And for $j=0$,
%
%\begin{eqnarray}
%
%Ja(0)&=&C(1-p_a) E[e^{-\gamma (\tau_a+\tau_\lambda)}]+\sum_{i=0}^{Q-1} V(i)  \left( \int \int_{0}^{\iy} e^{-\gamma(\tau_a+\tau_\lambda)}\frac{e^{-\lambda\tau_a}(\lambda \tau_a)^{i}}{i!}\mu_a e^{-\mu_a \tau_a} d\tau_a \right) + \nonumber \\
% &&+ V(Q)\int \int_{0}^{\iy}e^{-\gamma(\tau_a+\tau_\lambda)}\left(1-\sum_{i=0}^{Q-1}\frac{e^{-\lambda\tau_a}(\lambda \tau_a)^{i}}{i!}\right)\mu_a e^{-\mu_a \tau_a} d\tau_a\\
%&=&C(1-p_a)\frac{\lambda}{\lambda+\gamma} \frac{\mu_a}{\mu_a+\gamma}+\sum_{i=0}^{Q-1} V(i)  \left(\int  \int_{0}^{\iy} e^{-\gamma(\tau_a+\tau_\lambda)}\frac{e^{-\lambda\tau_a}(\lambda \tau_a)^{i}}{i!}\mu_a e^{-\mu_a \tau_a} d\tau_a \right) + \nonumber \\
% &&+ V(Q)\int  \int_{0}^{\iy}e^{-\gamma(\tau_a+\tau_\lambda)}\left(1-\sum_{i=0}^{Q-1}\frac{e^{-\lambda\tau_a}(\lambda \tau_a)^{i}}{i!}\right)\mu_a e^{-\mu_a \tau_a} d\tau_a\\
%\end{eqnarray}
Note that \rdo{\int^\iy_0t^ne^{-st} dt=\frac{n!}{s^{n+1}}}.
Thus, for $0 < j \leq B-1$ we have,
\begin{align}
 & J^u(j)=
c_u\frac{\mu_u}{\mu_u+\gamma}+\sum_{i=j-1}^{B-1} V(i)\frac{\lambda^{i-j+1}\mu_u}{(\gamma+\mu_u+\lambda)^{i-j+2}}+  \nonumber \\ %  \label{eq:exp2}\\
&+ V(B-1)\left(\frac{\mu_u}{\gamma+\mu_u}-\sum_{i=j-1}^{B-1}\frac{\lambda^{i-j+1}\mu_u}{(\gamma+\mu_u+\lambda)^{i-j+2}}\right), \label{eq:exp}
\end{align}
where $c_u=(1-p_u)$.
\iflong
\begin{comment}
 and 
\begin{align}
% & J^u(B-1)=C(1-p_u)\frac{\mu_a}{\mu_u+\gamma}+
%&\int_{0}^{\iy}e^{-\gamma\tau_a}\frac{e^{-\lambda\tau_a}(\lambda t)^{0}}{(0)!}\mu_a e^{-\mu_a \tau_a} d\tau_aV(Q-1)+\int_{0}^{\iy}e^{-\gamma\tau_a}\left(1-\frac{e^{-\lambda\tau_a}(\lambda t)^{0}}{(0)!}\right)\mu_ae^{-\mu\tau_a} d\tau_aV(Q)=\\
& J^u(B-1)=\frac{C(1-p_u)\mu_u}{\mu_u+\gamma}+V(B-2)  \frac{\mu_u}{\mu_u+\gamma+\lambda}+\\
&+V(B-1)\left(\frac{\mu_u}{\mu_u+\gamma}-\frac{\mu_u}{\mu_u+\gamma+\lambda}\right) \nonumber 
\end{align}
\end{comment}
\fi
The optimal value function is found by ${V(n)=\max_uJ^u(n)}$, $n=1, \ldots, B-1$. 
Note that 
the value function at the boundary condition  $n=B-1$ is obtained directly from the equations above, whereas  $V(0)$ is given by 
\begin{equation}
 V(0)=  \int_0^{\infty} e^{-\gamma t} \lambda e^{-\lambda t} V(1) dt  = \frac{\lambda}{\lambda+\gamma}V(1).\label{eq:exp1}
\end{equation}

%\iflong
%We now refer to example presented in~\ref{sub:exp}. 
%\fi
Observe that~\eqref{eq:exp} displays a quasi-closed form of the value function. This allows for convenient utilization of value iteration algorithm. In particular, by repetitious application of the recursive equations~\eqref{eq:exp}-\eqref{eq:exp1},  $V$ converges to the solution. 
Nevertheless, due to the summation in~\eqref{eq:exp} the value function of the SMDP at each state may depend on all other states. For this reason, a direct analysis of $V$ obtained from the SMDP formulation is cumbersome.  
% This particular observation makes the analyses of the value function intractable.
%In what follows we analyze a relay node system with exponentially distributed transmission times and prove the threshold structure of its optimal dynamic routing policy. 
To circumvent this challenge, we rely on  an alternative MDP formulation which is equivalent to the SMDP  presented above. The resulting Bellman equations are simple, hence analysis and identification of optimal policies of threshold type is plausible.

% The drawback of the MDP formulation which follows is expressed by the difficulty to compare the results of different cases associated with different transmission time distribution. This is because the value function of the MDP relate to any point in time, due to the memoryless property of exponential service times, while general SMDP only refers to the transmission end times. 
% Therefore, we will use a specific MDP form, which only incorporates states right after events (transmission ends and arrivals). The resulting Bellman equations are simple, hence analysis and identification of threshold policy is plausible.    
\subsubsection{MDP formulation for the exponential case}
%We now compare to the ``pure" MDP. 
%We now refer to example presented in~\ref{sub:exp}. 
%Observe that the general Bellman equations written by SMDP have a form where value function at each state depends on all other states. This particular property makes the analyses of the value function vector difficult.
%The objective of the alternative formulation is to analyze equivalent yet more simple representation.
%The pure MDP approach which is comparable to the SMDP solution involves state transitions at all changes in the number of tasks at the source. That is, both once a packet arrives and is served. Hence, we observe value function right after any arrival event \textit{and} after any transmission event. Clearly, the latter values are comparable with the SMDP values.

%We now formally define the states and the corresponding value functions. 
% Recall that \rdo{n\in\{0,\cdots,B\}} stands for the number of chunks in the buffer. 
Next, we define the states of the MDP and their corresponding value functions. 
The transition diagram of the MDP is illustrated in Figure~\ref{fig:mdpmodel}.  
The definition of the state space is inspired by the MDP admission control example presented in~\cite[chpt. 11]{puterman1994markov}. 

Our goal is to leverage the Markovian structure of the problem when the times between all events are exponentially distributed. To this aim, we 
modify the state space to $\{ 0, 1, \ldots, B-1 \}  \cup \{ 0, 1, \ldots, B-1 \} \times \{a,b\}$.  
%Whereas under the general SMDP framework transitions occurred after every departure, 
In particular, we now consider 
transitions that occur \emph{after every departure or arrival}.  
States $(n)$, $n=0,\ldots, B-1$, are achieved after a departure (transmission completion), 
whereas states $(n,u)$, $n=0, \ldots, B-1$, $u \in \{a,b\}$ are achieved after arrivals.

\begin{figure*}[t]\center
\includegraphics[width=0.9\textwidth]{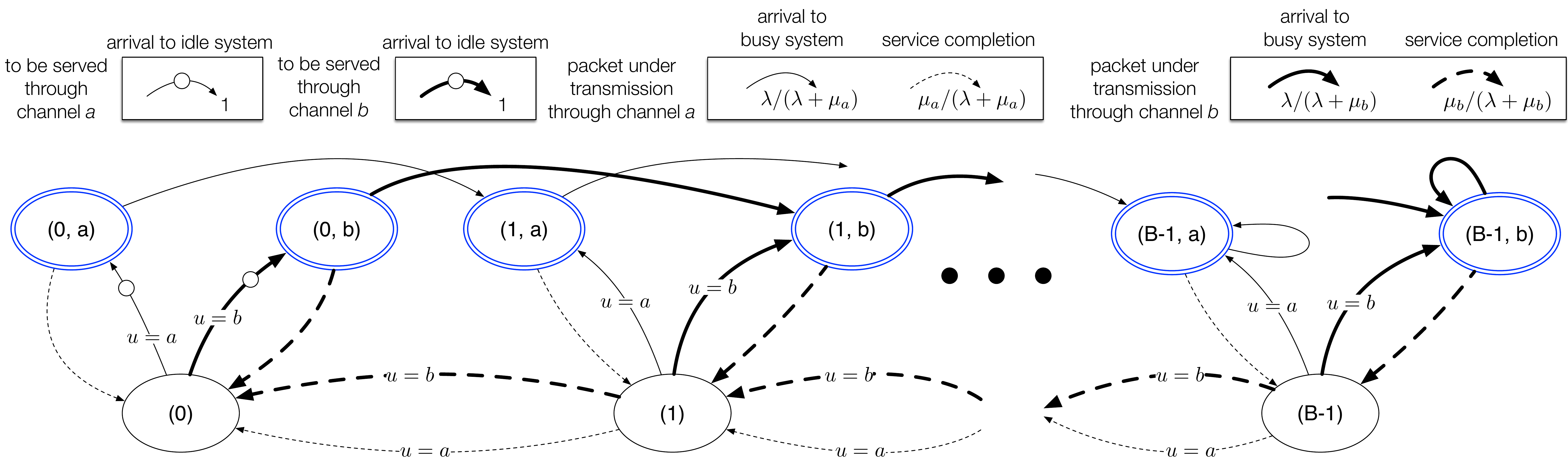}
\caption{MDP model of a system with buffer capacity $B$. Expected instantaneous  reward $c_a$ (resp., $c_b$) is received when action $a$ (resp., $b$) is taken.} \vspace{-0.2in}
\label{fig:mdpmodel}
\end{figure*}
The system transitions to state $(n)$ after a transmission completion that leaves behind $n$ packets in the buffer.  
As soon as the system reaches state $(n)$, $n=1,\ldots, B-1$, 
  the DM decides between transmitting the head-of-line packet through channels $a$ or $b$.  
  The mean
   residence time at state $(n)$, $n=1,\ldots, B-1$ is $1/(\mu_a+\lambda)$
    or $1/(\mu_b+\lambda)$, if actions $a$ or $b$ are chosen, respectively.  If a new packet arrives and encounters an idle system,
     the system transitions from state $(0)$ either to state $(0,a)$ or $(0,b)$, depending on whether action 
    $a$ or $b$ is chosen. The mean residence time at   state $(0)$ is $1/\lambda$.

    Next, we consider  a new packet that arrives to encounter a busy system.  %We assume that the arrival finds 
%    a total of $n$ packets in the buffer ,  and a packet
 %   being transmitted through channel $u$, $u \in \{a,b\}$.  
 Immediately after the arrival 
    the system transitions to state $(n',u)$, which accounts  for the packet being transmitted, and where $n'=\min(B-1,n)$.  At state $(n',u)$ the DM
    does not take any actions, as the only 
    possibility  is to continue the ongoing transmission.  The mean residence time at state $(n',u)$ is $1/(\mu_u+\lambda)$.

Let $\mathcal{P}(s'|s,u)$  be the transition probability from state $s$ to state $s'$, given that action $u$ was taken.
Then, for $n=1, \ldots, B-1$, the $\mathcal{P}(s'|s,u)$ is captured as follows.

% $\mathcal{P}(s_1|s_0,u_0)$ 

%\begin{align}
%  & \mathcal{P}(s_1|s_0,u_0) = \\

$\mathcal{P}(s'|s,u)=$
\begin{subequations} \label{ps1s0}
      \begin{empheq}[left={  = \empheqlbrace}]{align}
&1, && s = (0), s'=(0,u) \label{eq:sub1} \\
&\frac{\lambda}{\lambda+\mu_v}, &&  s = (n), s'=(n,v), v=u \textrm{ or }\;   \label{eq:sub2} \\
&& &s = (n-1,v), s' = (n,v)  \label{eq:sub3}  \\
& \frac{\mu_v}{\lambda+\mu_v},  &&   s=(n), s'=(n-1), v=u \textrm{ or } \;    \label{eq:sub4} \\
&&  &   s = (n,v), s'=(n) \textrm{ or } \label{eq:sub5}\\
&& &  s = (0,v), s'=(0)  \label{eq:sub6}
    \end{empheq}
\end{subequations}
%  
% \end{align}
\begin{comment}
1, & $ s_0 = (0), s_1=(0,u_0) $
$\frac{\lambda}{\lambda+\mu_u},$ & $ s_0 = (n), s_1=(n,u_0), u=u_0 \textrm{ or } $
& $s_0 = (n-1,u), s_1 = (n,u) $
$ \frac{\mu_u}{\lambda+\mu_u}, $ & $ s_0 = (n,u), s_1=(n) \textrm{ or } $
&$ s_0=(n), s_1=(n-1), u=u_0  $
0, &  $\textrm{otherwise}  $
\end{numcases} \label{ps1s0}
\end{comment}
%where $n=1, \ldots, B-1$.   

 A decision must be made when the system transitions to  state $(n)$, $n=0,\ldots, B-1$.    
% As we consider a non-preemptive system, once a transmission begins  it must be carried over.  
%  Therefore, decisions made at  
% states $(n,u)$ are inconsequential, and  
% transitions~\eqref{eq:sub3},~\eqref{eq:sub5} and~\eqref{eq:sub6} do not depend on $u$. 
Note that decision $u$ taken at  state $s=(n)$ immediately impacts the upcoming system state   
through transitions~\eqref{eq:sub1},~\eqref{eq:sub2} and~\eqref{eq:sub4}.
Transition~\eqref{eq:sub1} occurs after an arrival to an empty system, 
whereas transitions~\eqref{eq:sub2} and~\eqref{eq:sub3} occur after arrivals to a busy system. 
% A decision is made.
Transition~\eqref{eq:sub2} occurs after an arrival preceded by a  transmission completion, when decision $u$ had been made.  % which is currently acting according to $u$. No decision is made.
Transition~\eqref{eq:sub3} occurs after an arrival preceded by another arrival, 
during a transmission through channel $v$. 
Transitions~\eqref{eq:sub4}, ~\eqref{eq:sub5} and ~\eqref{eq:sub6} correspond to transmission completions.  
Transition~\eqref{eq:sub4} occurs after a transmission completion which was preceded by another transmission completion, 
when decision $u$ had been made. 
Finally, transitions~\eqref{eq:sub5} and ~\eqref{eq:sub6}
  occur  after transmission completions preceded by an arrival to a busy (resp., empty) system. 
% The previous state was associated with arrival to the system acting according to $u$. A (new) decision is made.
% Finally, transition corresponds to a transmission completion preceded by an arrival to an empty system.
%which leaves the system empty.
%  packet in the system. The system becomes empty till the next packet arrives and no decision is made till then.
%Note that~\eqref{ps1s0} depends on $u$ only at state transitions in~\eqref{eq:sub2} and~\eqref{eq:sub4}.
%At state transitions in~\eqref{eq:sub1},~\eqref{eq:sub3} and~\eqref{eq:sub5} are independent of $u$. which means that
 %the DM does not need to take any actions at those states.
Note that the at states $(n,a)$ and $(n,b)$ 
the variable $n$  does not account for  the arriving packet, which  will be admitted to the system   
in  case the buffer is not full. 
%Hence, there are $3$ states for each state of the buffer $n$, where $n\in0,\cdots,\BS$, denoted by $\{n,(n,d),(n,a)\}$. The value functions corresponding to the states are defined above. 
%Denote \sdo{V^{b,A}} and \sdo{V^{a,A}} the value functions at arrival events, at the moments of direct and relay transmission being active, correspondingly. There are no decisions made at the arrivals. Denote \sdo{V_{d,dep}} and \sdo{V_{r,dep}} values after transmissions are finished and the new transmission is decided to be performed on direct and relay channel, correspondingly. 

Instantaneous rewards are accumulated  once  transmissions are finished. Equivalently,  
 such rewards are added at the beginning  of a transmission, multiplied by the corresponding expected discount.  
 In what follows, we let $c_a$ and $c_b$ be the expected instantaneous reward received when actions $a$ and $b$ are taken,
 respectively.

%\subsubsection{MDP Bellman equations for the exponential case}

   Next, we introduce the value functions and Bellman equations which characterize the solution 
   of the MDP model. \ifArx Recall that  state  is achieved right \textit{after a departure} 
   (transmission completion) which leaves $n$ packets at the buffer,
$n \in \{0, 1, 2,\ldots, B-1\}$. \fi
The value function for state $(n)$ is denoted by $V_n$. 
\ifArx Recall also that 
state $(n,u)$ is achieved right \textit{after an arrival} to a system with $n$ packets,
$n \in \{0, 1, 2,\ldots, B-1\}$, 
 including the one which is currently being transmitted on path $u$, $u \in \{a,b\}$.  If $n=0$, the transmission of the
 arriving packet is immediately started through $u$.  \fi  
  The  value function corresponding to state $(n,u)$ is denoted by $V^{u,A}_n$ (``A" stands for ``arrival").
   
\begin{comment}  
    We start by summarizing the MDP state space and
   presenting the  value functions corresponding to each state.
   
   % To each state of the MDP there is a corresponding value, which is captured by 
   % the \emph{value function}.  

a) state $(n)$ is achieved right \textit{after a departure} (transmission completion) which leaves $n$ packets at the buffer,
$n \in \{0, 1, 2,\ldots, B-1\}$. The 
corresponding value function is denoted by $V_n$;

b) state $(n,u)$ is achieved right \textit{after an arrival} to a system with $n$ packets,
$n \in \{0, 1, 2,\ldots, B-1\}$, 
 including the one which is currently being transmitted on path $u$, $u \in \{a,b\}$.  If $n=0$, the transmission of the
 arriving packet is immediately started through $u$.  
  The corresponding value function is denoted by $V^{b,A}_n$ (``A" stands for ``arrival").
\end{comment}  
  
%\item state $(n,u,D)$, achieved right after a departure (transmission completion) which leaves $n$ packets at the buffer,
%$n \in \{1, 2,\ldots, N-1\}$,  and the next transmission is made through $u$,  $u \in \{a,b\}$.  
%  The corresponding value function is denoted by $V^{b,D}_n$ (``D" stands for ``departure").  
% \item state $\{n,b\}$, right after the arrival to the system which previously had $n$ packet, including the one which is currently being transmitted on path \textit{b}. The corresponding value function denoted $V^{a,A}_n$.
% \end{itemize} 

We are mainly interested in the values of states $(n)$, i.e.,  our main goal is to obtain  $V_n$, 
$n \in \{0, 1, 2,\ldots, B-1\}$,   
because the decisions are only made in these states and because they are \textit{comparable with the SMDP values}.

Let \rdo{\delta_a=(\mu_a+\lambda+\gamma)^{-1}}, \rdo{\delta_b=(\mu_b+\lambda+\gamma)^{-1}}, 
\rdo{\beta_b=(\mu_b+\gamma)^{-1}},\rdo{\beta_a=(\mu_a+\gamma)^{-1}}\iflong and \rdo{\bar\delta=(\gamma+\lambda)^{-1}}\fi. 
%We write next Bellman equations for the value function for states at arrival events.
In what follows it will always hold $u\in\{a,b\}$. 
The Bellman equations define the operators \rdo{\A_{a}} and \rdo{\A_{b}}, which act in the space of 
function from \rdo{\{0,\cdots,B\}} to $\R$. Write 
%\small
\begin{equation} 
 V^{u,A}_{n}=\mu_u \delta_u V_{n}+\lambda \delta_u V^{u,A}_{n+1} = \A_{u}V^{u,A}_{n} \label{eq:lp104} 
\end{equation}% \\
%& V^{a,A}_{n}=\mu_a \delta_a V_{n}+\lambda \delta_a V^{a,A}_{n+1} = \A_{a}V^{a,A}_{n},\nonumber}
%\normalsize
%The corresponding operators are given by:
%\fal{-1}{
%& {V}^{b,A}=\A_{d}V^{b,A},
%\;\;\; {V}^{a,A} = \A_{r}V^{a,A}\label{eq:pth7}
%}
%where superscript "D" stands for "departure". 
%Write value functions and the corresponding operators for both decision types. %\iflong The maximization is performed as follows
%\rals{-5}{
% V_{n}&=\max(V^{b,D}_{n},V^{a,D}_{n}) \nonumber\\
% &=\max((\T_b V_{n},\T_a V_{n}) =\T V_{n},%\label{lastbellman}
%}\else \fi
Denote by \rdo{\T,\T_a,\T_b} the operators acting on the space of functions from \rdo{\{0,\cdots,B\}} to $\R$. 
When applied to state $n$, these operators yield the following equation for $u\in\{a,b\}$: %\sdo{\T_a\text{ and }\T_b} defined by the following Bellman equations
\begin{equation}
%& V^{b,D}_{n}=\T_b V_{n}=\mu_b\delta_b V_{n-1}+\lambda\delta_b V^{b,A}_n+(1-p_b) \beta_b \nonumber \\
V^{u,D}_{n} = \T_u V_{n}=\mu_u \delta_u V_{n-1}+\lambda \delta_u V^{u,A}_n+(1-p_u)  \beta_u  \label{eq:lp113}
\end{equation}
The maximization over the available actions is  performed after transmission completions, 
at states $(n)$, $n=1,\ldots, B-1$, when $n$ packets are left in   the buffer, 
\begin{equation} 
V_{n}=\max(V^{b,D}_{n},V^{a,D}_{n}) =\T V_{n}=\max(\T_b V_{n},\T_a V_{n}).
\end{equation}
At the buffer limit boundary \rdo{\BS} we have %(see Figure~\ref{fig:beforeuniq1addstate}):
\begin{equation} 
 V^{u,A}_{B}=\mu_u \delta_u V_{B-1}+\lambda \delta_u V^{u,A}_{B}
\end{equation}
%& V^{a,A}_{B}=\mu_a \delta_a V_{B-1}+\lambda \delta_a V^{a,A}_{B}, %\label{eq:pth13}
and at the empty buffer, 
\begin{equation} \label{eq:pth14}
 V^{b,D}_{0}=V^{a,D}_{0}=\max(\lambda \bar\delta V^{b,A}_{0},\lambda \bar\delta V^{a,A}_{0}). %\label%+(1-p_b)*(C)*(miu_b+l1)*delta_b;
\end{equation}
Arrivals that find an empty buffer are subject to the effect of the DM current decision, 
\begin{equation} \label{eq:bellman2}
 V^{b,A}_{0}=V^{a,A}_{0}=%\nonumber \\ 
\max_{u\in\{a,b\}}(\mu_u\delta_u V^{u,D}_{0}+\lambda\delta_u V^{u,A}_1+c_u)
\end{equation}
%\mu_a \delta_a V^{a,D}_{0}+\lambda \delta_a V^{a,A}_1+c_a) \nonumber \label{eq:115}}
Note that $V^{b,A}_{0}=V^{a,A}_{0}$ holds because at state $(0)$ no packet is currently being transmitted.
%See that at states $(n,u)$, $u\in\{a,b\}$, there are $n+1$ packets, whereas at states $(n)$ and $(n,u,D)$ there are $n$ packets.
The derivation of  equations~\eqref{eq:lp104}-\eqref{eq:bellman2} from the process model is presented in Appendix~\ref{app:pBE}.
%\clearpage
%\pagebreak

\begin{table}[t]
     \vspace{0.1in}
 \center
\begin{tabular}{l|l}

\hline
variable & description \\ %  & comment \\

\hline
%${\bS}$ & state space of MDP & $|{\bS}|=3\BS+2$ (see Fig~\ref{fig:beforeuniq2addstate} for $\BS=2$) \\
%$s$ & state & \\%$\sigma=(n,a)$ \\
 $c_u$ & expected instantaneous reward associated to decision $u$.  \\ \hdashline %& weighted by the expected discount at the end of transmission \\
%        $c_b$ & expected instantaneous  reward associated with decision ``b".  \\  \hdashline%  & weighted by the expected discount at the end of transmission \\
% $n$ & number of packets in the waiting line,  not counting the packet \\
% &  being transmitted,   $n \in \{ 0, 1, \ldots,  B-1 \}$ \\
% $u$ & current active channel, $u \in \{a,b\}$ \\
$(n,u)$ & state right after arrival, when current active transmission \\
&  is through $u$ and $n$ packets are found by the arrival.  \\
 % & Note: $n$ is the number of packets in the buffer, \\
 % & accounting for the packet  being transmitted \\
 % & \emph{but not for the latest arrival} \\    
%$(0,u)$ & state right after arrival to an empty buffer \\%  &   \\
%     &  (arriving packet will be transmitted by $u$),   $u \in \{a,b\}$ \\
    $(n)$ & state  following a transmission completion, when $n$ packets \\
    &  are left in the buffer. Decisions are made at these states.     \\   \hdashline%
   % $(n,u,D)$ & state  following a transmission completion,  \\
   % & after decision of transmitting through $u$,  $u \in \{a,b\}$  \\
   % &  $n \in \{1,\ldots,B-1 \}$    \\  
   % & Note: $n$ is the number of packets in the buffer, \\
   % & accounting for the packet being transmitted \\    \hdashline    
% $N$ & random variable characterizing number  & $N \in \{0,1,\ldots,B \}$  \\
% & of packets in queue, i.e., & \\
% & waiting queue + under transmission & \\ 
%             & (number of packets tracked  immediately after   & \\
%             & arrivals and departures, as seen by the arrivals  & \\
%             &  and left behind by the departures) & \\             \hdashline
    $V^{u,A}_n$ &  value  function at state $(n,u)$. \\ % ,  $u \in \{ a, b\}$, $n \in \{ 0, \ldots, B-1 \}$ \\
                 $V_n$ &  value function at state $(n)$,   $V_n=\max\{V^{a,D}_n,V^{b,D}_n\}$. \\
 $V^{u,D}_n$ &   state-action value function for action $u$ at state $(n)$. \\  \hdashline%
 %& $u \in \{a,b\}$,   at state $(n)$, $n \in \{ 0, \ldots, B-1 \}$ \\ \hdashline
        % & after reaching  state $(n,t)$ & \\
         % $V^{b,D}_n$ & value  function associated with action $b$ &  $n \in \{ 0, \ldots, B-1 \}$ \\
       %  & after reaching  state $(n,t)$ & \\
          %  $V^{b,A}_n$ &  value function at state $\{n,b\}$  & $n \in \{0, \ldots, B \}$ \\  
    $\A_{u}$ & operator applied over  arrivals (acts on $V^{u,A}$). \\ %  & $U$ \\
  %  & when the active transmission is through ``u", $u \in \{a,b\}$   \\
    $\T_{u}$ & operator applied over departures (acts on $V^{u,D}$). \\ % , with action $u$ \\
   % & transmission through ``u" (acts on $ V_n$).\\
    $\T$ &  operator that maximizes the outcome of $\T_u$ over $u$. \\
      \hline
     \end{tabular}
     \vspace{0.1in}
     \caption{MDP notation ($u \in \{a,b\}$, $n \in \{0,\ldots,B-1 \}$).} \vspace{-0.3in}  \label{tab:mdpspec}
     \end{table}
  %   \clearpage
%Let
%\fal{-1}{ \label{eq:deltandef}
%\Delta_0 &=& 0 \\
%\Delta_n &=& \nu_n - \nu_{n-1}, \qquad n=1, 2, \ldots, \BS \nonumber
%}
%and convene to \sdo{\nu_0=0}. 
%Observe that there exists an optimal threshold type in case $\nu_n$ is increasing in $n$, that is,
%\seq{-1}{V^{b,D}_{n}-V^{b,D}_{n}\ge 0, \qquad n=1, 2, \ldots, \BS}{eq:pth1}
In Section~\ref{sec:props} we will use the MDP formulation to identify the threshold type structure of the optimal policy.

\fi
% \vspace{-0.25in}
\subsection{{Deterministic transmission times and known packet sizes}} \label{sec:det}
\normalsize Consider a source which samples the packet size before the transmission.
Consider two possible sizes, denoted by $k_1$ and $k_2$, both taking in the buffer exactly one slot. %, i.e. case $2)$ in section~\ref{sec:phy}. 
We assume equal rewards for both sizes.
Then, the state is given by \rdo{s=\{n,k\}}, \rdo{n\in\{0,\cdots,B\}} and \rdo{k\in\{k_1,k_2\}}.
\begin{comment}
The packet is ready to be transmitted once the newly arrived chunk indicates that it is the last one in the packet.
\end{comment}
We assume that packet size dynamics is given by a discrete Markov Chain with transition probabilities $q(k'|k)$. 
%For simplicity we also assume $\bA=\{a,b\}$, that is, the packet, if ready, has to be transmitted.
Denote the deterministic transmission time of packet of size $k$ as $\tau^{\pi(n,k)}$.
The packet loss probabilities are $p_{a}$ and $p_{b}$. %for the size $k$ of a packet to be transmitted.
%The instantaneous rewards are $C(1-p_{r,k})$ and $C(1-p_{d,k})$, correspondingly, which are further weighted by the accumulated delay $\beta_{u,k}$ at the end of the transmission which is about to start.
%\iflong We have
%\rals{-5}{
%& J^u(i,k)=C(1-p_{a,k})\beta_{u,k}+ \\
%&\sum_{k'}\sum_{j=i-1}^BV(j,k)\int_0^\iy e^{-\gamma t}\varrho(j|i,t)q^u(k'|k)\delta(t-\tau_{u,k})dt=\\
%&\;\; C(1-p_{u,k})e^{-\gamma\tau_{u,k}}+\\
%&\sum_{k'}\sum_{j=i-1}^\BS e^{-\gamma\tau_u}\varrho(j|i,\tau_{u,k})q^u(k'|k)V(j,k'),} 
%\else
Denote $s_0=(i,k),\;s_1=(j,k')$.
Then,  %and $\beta^{\pi(n,k)}=e^{-\gamma\tau^{\pi(n,k)}}$ we have %Write equation for a packet of size $k$ at the head of the source's queue.
%\vspace*{1pt}
\begin{equation*}  % =\E r^\pi(i,k)+
J^{\pi}_1(i,k)  = \sum_{k'}\sum_{j=i-1}^\BS e^{-\gamma\tau^{\pi(s_0)}}Q((s_1)|s_0,\pi(s_0))V(j,k')
\end{equation*} 
where  $Q((s_1)|s_0,\pi(s_0))=q(k',k)\varrho(j|i,\tau^{\pi(s_0)})$. See that  
\(c_{\pi(i,k)}=\E r^{\pi(i,k)}=(1-p_{\pi(i,k)})e^{-\gamma\tau^{\pi(n,k)}}\).
%where $k$ is the packet size at the head of the queue and Poisson distribution for $\varrho$ is substituted. 
%The decision space consists of subspaces where each subspace is uniquely corresponds to a packet size.
Hence, the value functions are
\rdo{ V(i,k_m)=\max_\pi\{J^{\pi(i,k_m)}\}}, $m\in\{1,2\}$. %;  V(i,k_2)=\max_\pi\{J^\pi(i,k_2)\}}
%& V((i,B))=TV(i,B)=\max\{V((i,B),r),V((i,B),d)\}
\iflong
Finally, the boundary condition, then the buffer is empty is
\rals{0}{
V(0,k)=  \E e^{-\gamma t_\lambda}V(1,k)=\sum_{k'}\frac{\lambda}{\gamma+\lambda}\big(q(k'|k,0)V(1,k)\big).
}
\fi
%and
%{\small{\[
%V(0,B)=\frac{\lambda}{\gamma+\lambda}\big(q_0(A|B,0)V(1,A)+q_0(B|B,0)V(1,B)\big),
%\]}}

\vspace{-0.2in}

\subsection{{Gilbert-Elliott   channel, uniformly distributed transmission times}}\label{sec:uniform}
%Consider uniform transmission time distributions (e.g. 
Assume the packet sizes cannot be sampled, but are known to have a uniform distribution over all channels.
Consider a Gilbert-Elliott (G-E) channel with two states, Good and Bad, denoted by $\mathcal{G}$ and $\mathcal{B}$.
 For simplicity, we assumed a similar states for the \textit{entire medium}, that is
  $h\in\{h^{\mathcal{G}},h^{\mathcal{B}}\}$, such that $h^{\mathcal{G}}=(h^{\mathcal{G}}_a,h^{\mathcal{G}}_b)$, 
  $h^{\mathcal{B}}=(h^{\mathcal{B}}_a,h^{\mathcal{B}}_b)$, where $a,b$
   stand for the two operating points, 
 i.e., channel ``a'' (e.g.,  the relay channel) and channel ``b'' (e.g.,   the direct channel). We assume that channel dynamics can be expressed by discrete Markov Chains.
%\iflong Namely, both relay and direct routes are fully dependent and can simultaneously be in one of the states. \fi  That is, the channel state refers to the entire medium. 
The channel state is sampled prior to each upcoming transmission,
and is modeled as part of the state space.  Hence, the state is given by $s=\{n,h\}$, \rdo{n\in\{0,\cdots,\BS\}} and \rdo{h\in\{{\mathcal{G}},{\mathcal{B}}\}}.  
The packet loss probabilities can be calculated right before the current transmission slot. According to the number of possible channel states, there are two possible packet loss probability ordered pairs: $(p_a^{\mathcal{G}},p_b^{\mathcal{G}})$ and $(p_a^{\mathcal{B}},p_b^{\mathcal{B}})$.
 We assume these probabilities are calculated from the BER which corresponds to the particular G-E channel state.  
Denote $u=\pi(n,h)$.
%, which corresponds to the Gilbert–Elliott (GE) channel.
The transmission time $\tau^{a,h}$ (resp. $\tau^{b,h}$) over channel ``a''  (resp. channel ``b'')
  is uniformly distributed. The uniform distributions intervals are given by \rdo{[\alpha_{h_u},\beta_{h_u}]}. 
%We assume that all packets in the buffer  occupy exactly one slot and are equally rewarded.
The channel transition probabilities, are denoted by $p(h'|h)$.
Notation is summarized in Figure~\ref{fig:gel}.

\begin{figure}[h!]
\center
\includegraphics[width=0.4\textwidth]{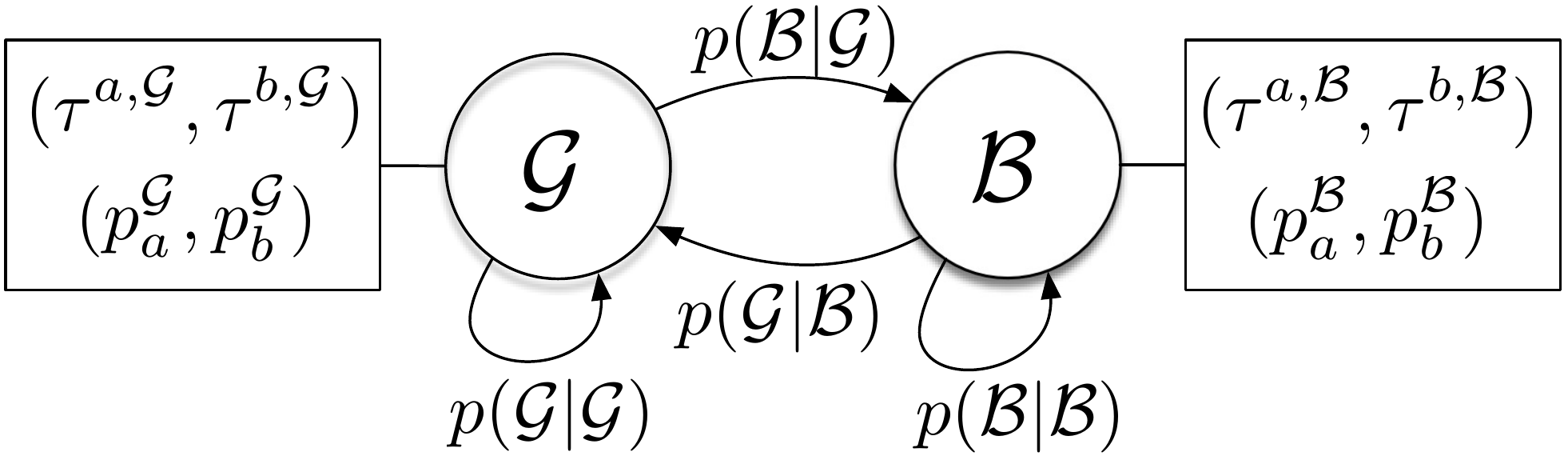}
\caption{Gilbert-Elliott model} \label{fig:gel}
\end{figure}

 %, depend on the action $u\in\{r,d\}$, where relay (resp. direct) transmission corresponds to the transmission mode "a" (resp. "b").  
Denote $s_0=(i,h), \;s_1=(j,h')$, $h,h'\in\{{\mathcal{G}},{\mathcal{B}}\}$. Then, 
\begin{equation*}
J^{u}_1(s_0)=\sum_{j=i-1}^B\sum_{h'} V(j,h')\int_{\alpha_{h_u}}^{\beta_{h_u}} e^{-\gamma t}Q(s_1|s_0,\pi(s_0))dt
\end{equation*}
where \rdo{Q(s_1|s_0,\pi(s_0))=({\beta_{h_u}-\alpha_{h_u}})^{-1}\varrho(j|i,t)p(h'|h)}, 
%$C(1-p_{u,h})\beta_{u,h}$
%Next, assume the packet transmission time $\tau_d$ over direct channel is uniformly distributed in $[a_d,b_d]$.
%\begin{align*}
%& V(j,d)=C(1-p_d)+\frac{1}{b_d-a_d}\sum_{i=j-1}^Q V(j)\int_{a_d}^{b_d} %e^{-\gamma\tau_d}\varrho(i|j,\tau_d)d\tau_d
%\end{align*}
%To solve the integration, Poisson distributions are substituted, and well known integration \sdo{\int x^n e^{cx}\; \mathrm{d}x = e^{cx}\sum_{i=0}^n (-1)^{n-i}\,\frac{n!}{i!\,c^{n-i+1}}\,x^i} is applied.
and %\sdo{\beta_{u,s}} is the discount which is given by 
$({\beta_{h}-\alpha_{h_u}})\E r^u=(1-p_u)\int_{\alpha_{h_u}}^{\beta_{h_u}}e^{-\gamma t}dt=(1-p_u){e^{-\gamma(\beta_{h_u}-\alpha_{h_u}) }}$.
%Substitute~\eqref{be2}, $\varrho(i|j,\tau_d)=\frac{e^{-\mu \tau_d}(\mu \tau_d)^{i-j+1}}{(i-j+1)!}$
%\begin{align*}
%& V(j,d)=C(1-p_d)+\frac{1}{b_d-a_d}\sum_{i=j-1}^{B-1} V(i)\int_{a_d}^{b_d} e^{-\gamma\tau_d}\frac{e^{-\lambda \tau_d}(\lambda \tau_d)^{i-j+1}}{(i-j+1)!}d\tau_d\;+\\
%& \frac{1}{b_d-a_d}V(Q)\int_{a_d}^{b_d} e^{-\gamma\tau_r}\varrho(B|i,\tau_r))d\tau_d
%\end{align*}
Note that the probability to have full buffer after end of transmission is given by
$\varrho(B|i,t)=(1-\sum_{i=j-1}^{B-1}\varrho(j|i,t))$.
%\begin{align*}
%& V(j,r)=C(1-p_r)+\frac{1}{b_r-a_r}\sum_{i=j-1}^{Q-1} V(i)\frac{\lambda^{i-j+1}}{(i-j+1)!}\int_{a_r}^{b_r} e^{(-\gamma-\lambda)\tau_r}( \tau_r)^{i-j+1}d\tau_r\;+ \\
%&\frac{1}{b_r-a_r}V(Q)\int_{a_r}^{b_r} e^{-\gamma\tau_r}(1-\sum_{i=j-1}^{Q-1}\varrho(i|j,\tau_r))d\tau_r
%\end{align*}
Finally, the value function for state $s$ is given \rdo{V(s)=\max_u\{J^u(s)\}}.

% \vspace*{-8pt}

\section{Structure of optimal policies}\label{sec:props}
%It is important to understand
The structure of the optimal policy has a particular importance, in the sense that it can facilitate assessment of resources needed for the policy implementation at wireless nodes. For the system with large state-space, structural properties can be exploited by learning algorithms in order to significantly reduce the complexity of optimal policy search. %This is especially useful .
For example, once the policy is proven to possess a \textit{threshold} structure, the data to hold for the policy (in the corresponding dimension of a state space) is reduced to a single scalar.
Moreover, the configuration of similar systems can be analytically or heuristically based on the existing one, e.g. by means of reinforcement learning aimed to policy improvement.
We aim to identify threshold policies for the SMDP models and solutions presented above. For the exponential case, we analytically prove the threshold property. We finally compare by simulations the thresholds associated with other transmission time distributions. 
To this end, we state our main analytical result:
%We now refer to example presented in~\ref{sub:exp}. 

\begin{theorem}\label{prop:1} 
	The problem with exponentially distributed transmission times modeled by MDP is solved by the optimal policy of a threshold type. Namely, there exists a unique threshold $t$, $0\leq t\leq B$, such that the optimal policy is to transmit via path ``a" for all
	 states where $n\leq t$ and to transmit via path ``b" otherwise.
\end{theorem}
\iflong
By the equivalence of the value functions at departures in MDP and SMDP formulations trivially the following holds.
\begin{corr}
	The exponential problem modeled in section~\ref{sec:mdpmodel} by SMDP is solved by the optimal policy of threshold type.
\end{corr}
\else
By the equivalence of the value functions at departures the proposition holds both for MDP and SMDP. 
\fi
\iflong
%We will use the following definition. %of concavity of function \sdo{U_n}  and, in addition, we define inter-concavity of 
%\begin{definition}[Inter-concavity] \label{def:interconcavity}
%$U_n$  is inter-concave with respect to $Y_n$ if
%\rals{-1}{
%%& U_{n}-U_{n-1}\geq U_{n+1}-U_{n}\text{  concavity} \\
%& U_{n}-U_{n-1}\geq Y_{n+1}-Y_{n}\text{  inter-concavity}}
%\end{definition}
%We say that $U_{n}$ is \textit{concave} if it is inter-concave with respect to itself.

To this end, let $\calS$ be a set where each of its  elements is a five-tuple of \sdo{\BS}-dimensional vectors denoted by \rdo{\{U,U^{b,A},U^{a,A},U^{b,D},U^{a,D}\}} satisfying the following properties
%  \sdo{U,U_b,U_a} are  \sdo{\BS}-dimensional vectors. % $U(n)$, $n\in\{0,\cdots,\BS\}$ and one
% \sdo{(U^{(direct)}_{dep},U^{(relay)}_{dep})} is  a \sdo{2\BS}-dimensional vector,  and the four-tuples  
%\begin{enumerate}\label{enum1}
 %   \item 
    
    1) the difference \rdo{U^{a,D}_{n}-U^{b,D}_{n}} is non-decreasing in $n$, \rdo{n\in\{0,\cdots,\BS\}} 
	
	2)  \rdo{\{U^{b,D},U^{a,D},U^{b,A},U^{a,A}\}} are concave in \rdo{n\in\{1,\cdots,\BS\}},
	
	3) \rdo{\{U,U^{b,A},U^{a,A}\}} are  non-decreasing  in \rdo{n\in\{0,\cdots,\BS\}},
	%\item \rdo{U^{a,A}} is inter-concave with respect to $U^{b,A}$, and \rdo{U^{r,d}} is inter-concave with respect to $U^{d,d}$.

	4) \rdo{\{U,U^{b,A},U^{a,A}\}} have their slope bounded by some positive constant \rdo{K}, that is, \rdo{U_{n}-U_{n-1}<K}, \rdo{U^{a,A}_{n}-U^{a,A}_{n-1}<K} and \rdo{U^{b,A}_{n}-U^{b,A}_{n-1}<K}.
%\end{enumerate}
For the proof of the theorem we will need the following lemma.
\begin{lemma} \label{mainlemma1}
	% \shifrinA{this formulation is weaker than we prove. we do not need this condition. lets change it back.}
	The operators $\A_b$,$\A_a$,$\T$ preserve properties 1)-4).
	%\reqs{
	%(\T V_{dep}, \A_b V^{b,A}, \A_a V^{a,A}, (\T_{d} V_{dep}, \T_{r} V_{dep})) \in \calS
	%}
\end{lemma}
\ifArx
The proof of the lemma appears in Appendix~\ref{app:pr_lem}.
\else
The proof of the lemma appears in the on-line version of the paper~\cite{on-line}.
\fi
\else

Let $\calS$ denote the set of functions from \rdo{n\in\{1,\cdots,\BS\}} to $U^{r,d}_{n}$ and $U^{d,d}_{n}$ such that 

%\small
\req{-5}{U^{r,d}_{n}-U^{d,d}_{n}\text{ is increasing in $n$}.}{eq:propshort}
\normalsize
For the proof of the theorem we will need the following lemma.
\begin{lemma} \label{mainlemma}
	% \shifrinA{this formulation is weaker than we prove. we do not need this condition. lets change it back.}
	The operators $\A_b$, $\A_a$, $\T$ preserve property~\eqref{eq:propshort}.
	%\reqs{
	%(\T V_{dep}, \A_b V^{b,A}, \A_a V^{a,A}, (\T_{d} V_{dep}, \T_{r} V_{dep})) \in \calS
	%}
\end{lemma}
\fi
% \shifrinA{is good to make a proposition! i will make some fixes to this after we finish the lemma.}
\vspace*{-5pt}\begin{proof}[Proof of Theorem~\ref{prop:1}]
	%From the lemma, we know that once the corresponding operators $\T$, $\A_b$ and $\A_a$ are applied,
	%the result stays in $\calS$. By fixed point theorem and the contraction properties of the operators,
	%the value function which uniquely solves the Bellman equations~\eqref{firstbellman}-\eqref{lastbellman} is also in $\calS$.
	We rely on a well known result that operators associated with Bellman equation are contracting~\cite{puterman1994markov}\iflong, that is, using the maximum metric \sdo{\parallel U\parallel=\max_x|U(x)|} it holds \sdo{a\parallel U-W\parallel<\parallel\T U-\T W\parallel} for some \sdo{0<a<1}. Hence, the operators defined above are contraction mappings\else. \fi
	%We use the contraction mapping principle (see e.g. in~\cite[Theorem V.18]{reed1980methods}), which states the existence of fixed point of strict contraction mapping in a complete metric space. The set \sdo{\calS}
	\iflong, equipped with the metric
	\rdo{\rho(U;W) = ||U - W||}   in a complete metric space. \fi %The mappin
	%T : S ! S is a strict contraction, as shown in the above
	%lemma.
	Since $\calS$ is a complete metric space and the operators are strict contractions, they have corresponding fixed points (e.g.~\cite[Theorem V.18]{reed1980methods}).
	%Since the operators \sdo{\A_a, \A_b,\T} are strict contractions, they have corresponding \textit{unique} fixed points. 
	Now since $\calS$ is not empty (one can easily construct such functions; the technical details are omitted), the functions which are in $\calS$ and have the operators \rdo{\A_a, \A_b,\T}  applied on them, by lemma~\ref{mainlemma1} stay in $\calS$. By contraction, the repetitious application brings the result infinitesimally close to the fixed points of \rdo{\A_a, \A_b,\T}.   %That is, there
	%exists a unique U 2 S for which TU = U. 
	Recall that the value functions \rdo{V_{n}, V^{b,A}_{n},  V^{a,A}_{n}}  are
	the unique solution of \textit{all} functions, including those that in $\calS$, acting from \sdo{n\in\{0,\cdots,B\}} to $\R$, to the \textit{same} equations; (trivially, the mild conditions for uniqueness and existence, see e.g.~\cite[Chapter 6.2]{puterman1994markov}, apply). %of all
	%functions from f0; 1; : : : ;Bg to R. 
	As a result, \rdo{\{V, V^{b,A},  V^{a,A}, V^{b,D},  V^{a,D}\}} coincide with these fixed points and %by not-emptiness of $\calS$ and the uniqueness of the solution to the Bellman equation 
	they are in $\calS$. 
	\iflong
	In particular, 
	$V^{b,D}$ and $V^{a,D}$ possess property 
	4), which is equivalent of having at most one policy switch state. This proves the proposition.
	\else
	In particular, 
	$V^{b,D}$ and $V^{a,D}$ possess property~\eqref{eq:propshort}, which is equivalent of having at most one policy switch state. This proves the proposition.
	\fi
\end{proof}
\vspace{-0.2in}
\section{{Numerical results}} \label{sec:numer}
In this section, we report numerical results on the shape of the value functions obtained through value iteration.  Our goals are to 1) illustrate how different system parameters impact the performance of  threshold policies and 2) numerically investigate the optimality of multi-threshold  optimal policies for the Gilbert-Elliott channel.

The parameters chosen in the numerical experiments that follow in this section were 
 % simulations and below were not inspired by a concrete physical setting, but  were
  selected for illustrative purposes, and are set according to the experimental goals.
  
  % in order to exemplify the mathematical properties of the model.  
  
\subsection{{Value functions}} 

In Figure~\ref{fig3}, we compare the  value functions and threshold policies for channels associated with exponential, deterministic and uniform transmission times. The mean transmission rates were set to $\mu_a=9$ and $\mu_b=12$, under channels $a$ and $b$, respectively.   The support of the uniformly distributed transmission times was set to  $[\alpha, \beta]$, where $\alpha=0.2/{\mu_u},\beta=1.8/{\mu_u}$ and $u \in \{a,b\}$. We considered both a low load ($\lambda=\mu_a=9$) and a high load ($\lambda=\mu_b=12$) regime. Vertical lines show the thresholds where the policy determines a switch from $a$ to $b$.

\iflong\else\vspace*{0pt}\fi\begin{figure}[h!]
	\begin{center}
	\iflong \vspace{-10pt} \else\vspace*{-15pt}\fi
		\includegraphics[angle=0, width=0.4 \textwidth]{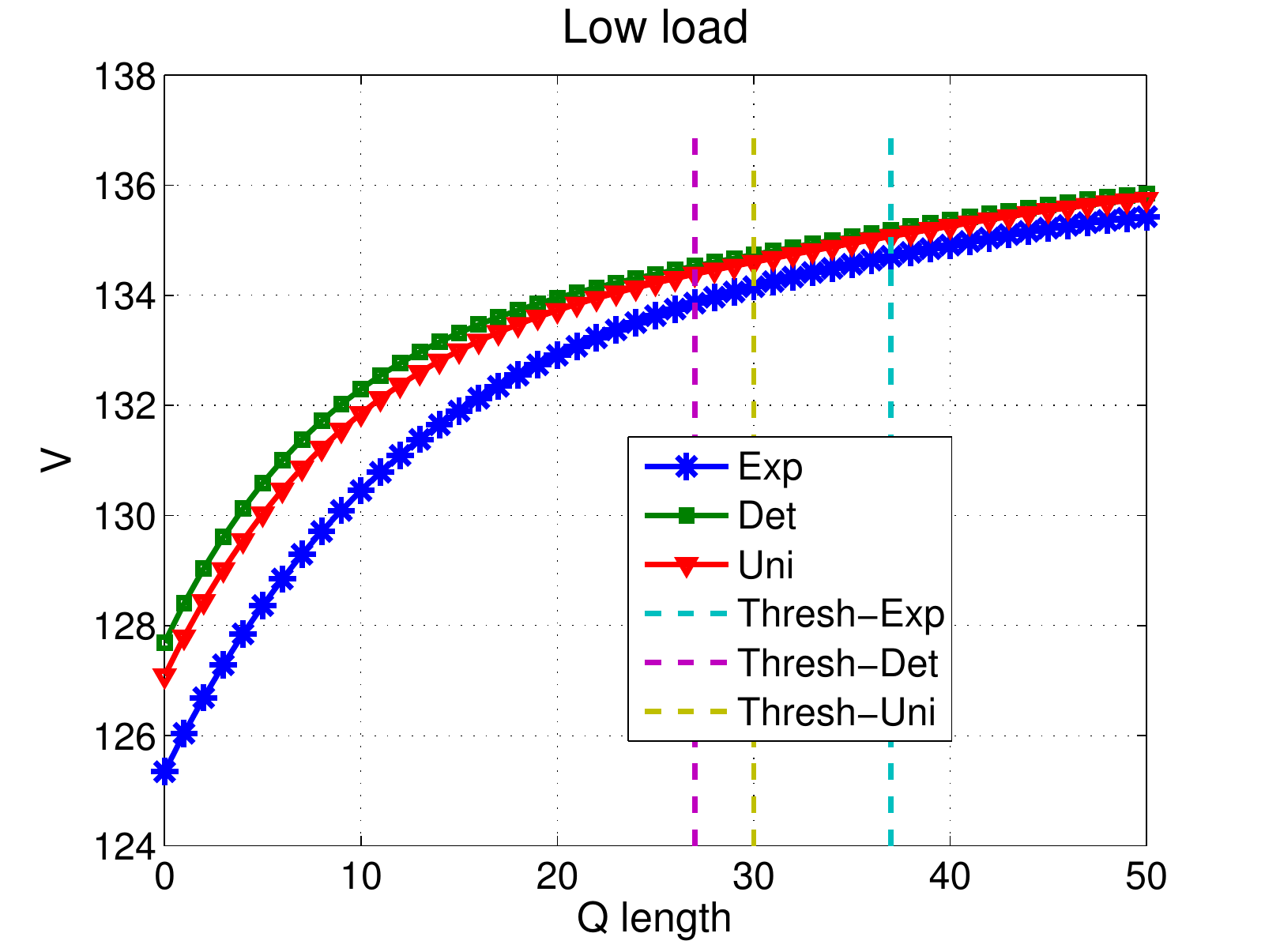} \\ 
		\includegraphics[angle=0, width=0.4 \textwidth]{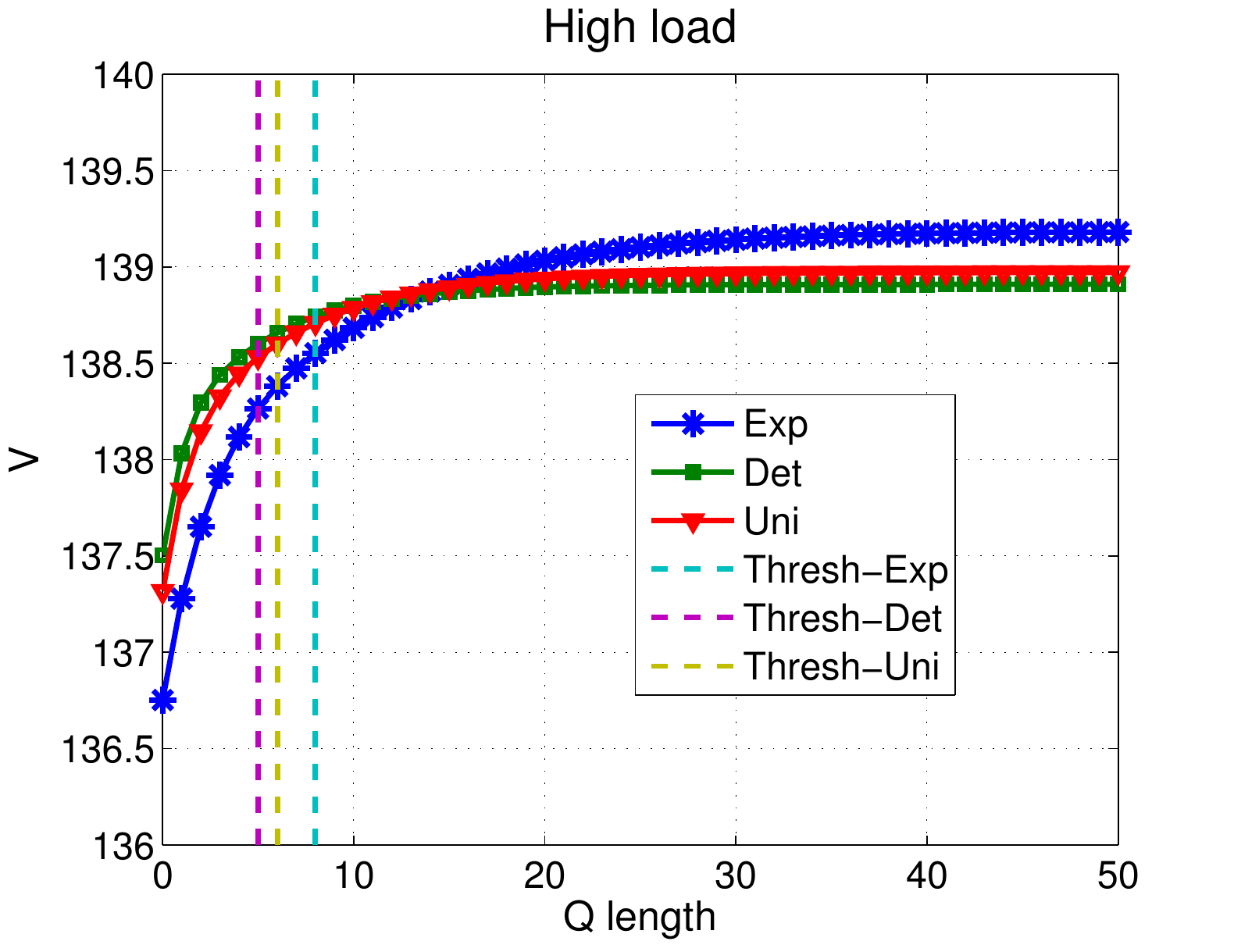}
		%\hfill
		\iflong \vspace{-10pt} \else\vspace*{-25pt}\fi
		\caption{\footnotesize{Value functions, $V(\cdot)$, at different buffer states,  $n$, 
		 when channel states are sampled from an i.i.d. model. 
		Average transmission rates and packet losses are given by $\mu_a=9, p_a=0.25$ and $\mu_b=12, p_b=0.42$, under  %For uniform $u_{u}=0.2\frac{1}{\mu_u},d_{u}=1.8\frac{1}{\mu_u}$. 
		high load ($\lambda=13$) and low load ($\lambda=9$).}\label{fig3} 
		% Packet loss values were $p_{a}=0.25,p_{b}=0.42$}
		% \vspace*{-15pt}
		}
	\end{center}
\end{figure}

Observe that under  high load the thresholds are significantly smaller than under low load. This is because  in the latter case it is important to avoid buffer underflows, which cause a reduction in system throughput.    
% Note that rejected packets also imply a mere unaccomplished reward, as long as no rejection fines are incurred. Hence, the rejections have no impact on the accumulating reward. Consequently, the values are at high load are higher.
For the high load, see that at the states close to $B$ the value function becomes nearly constant. This may be explained by the fact that in all these states the average time until the buffer  empties  is large. %Hence, at these states the penalty due to a potential non-realized  future discounted reward is negligible.
% which will be lost is  negligible.   

Also note that the numerical results validate our formal results on the concavity of the value function for the exponential case. In addition, the value function for the two other cases are observed to have a concave form, an observation which is interesting on its own.

%Value functions for Gilbert-Elliot channels are presented in figure~\ref{fig4}.
\iflong\else\vspace*{0pt}\fi\begin{figure}[h!]
		\vspace{-0.1in}
	\begin{center}
	\iflong\else\vspace*{-15pt}\fi
		\includegraphics[angle=0, width=0.4 \textwidth]{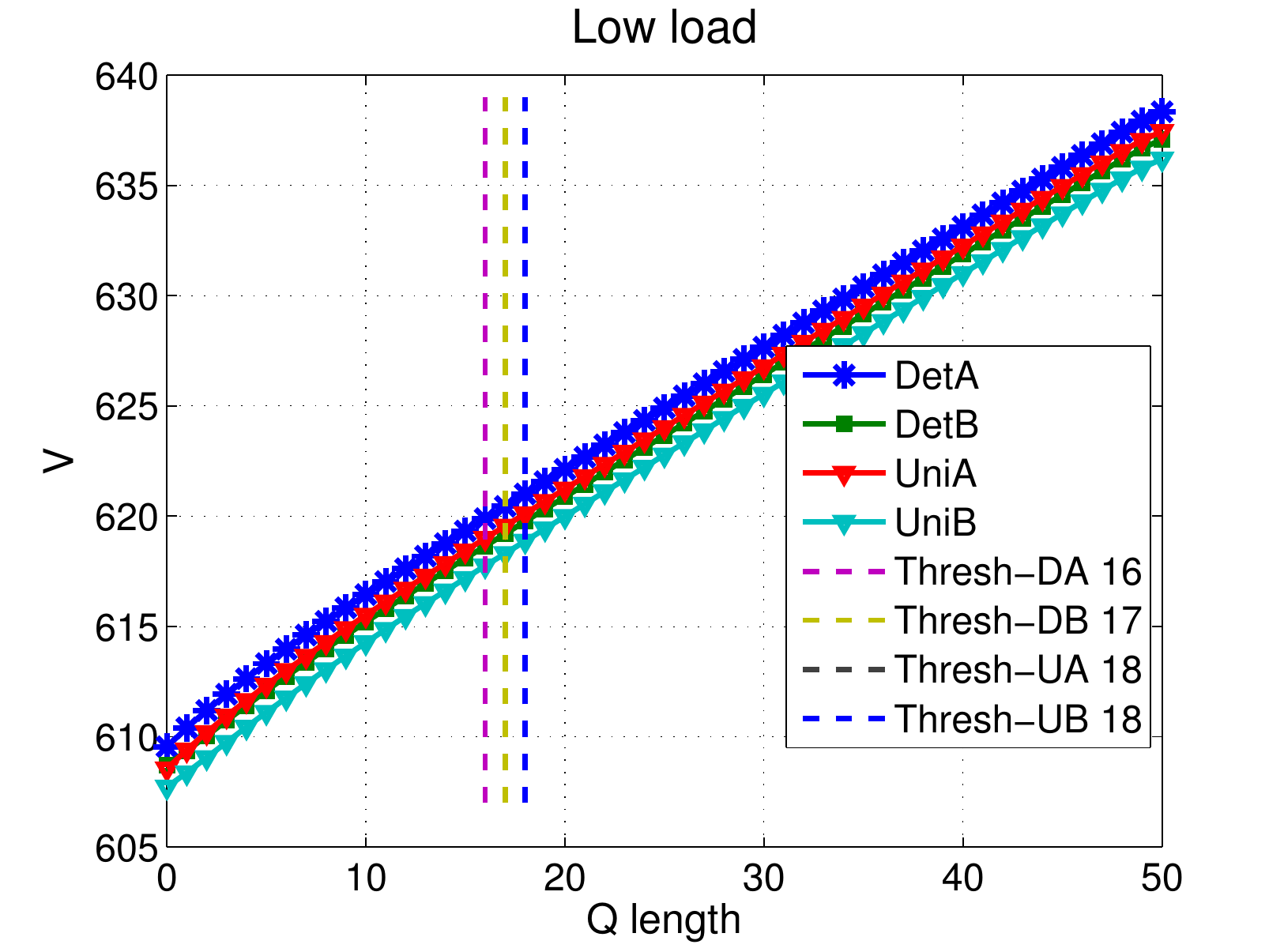} \\
		\includegraphics[angle=0, width=0.4 \textwidth]{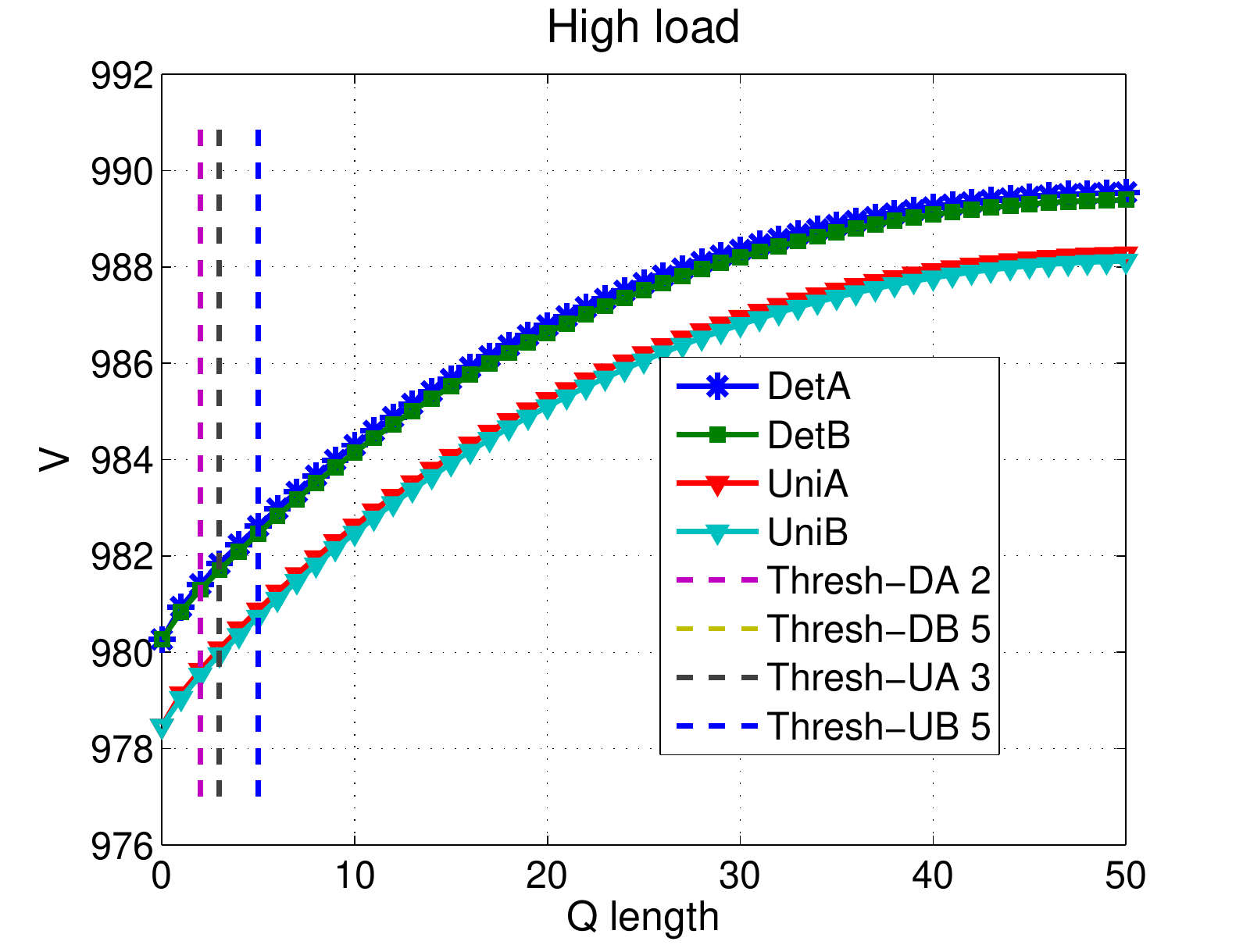}
		\vspace{-0.1in}
		%\hfill
				\iflong \else\vspace*{-25pt}\fi
	\caption{\footnotesize{ Value functions, $V(\cdot)$, at different buffer states, $n$, when channel is characterized by a Gilbert-Elliott model. Parameters are given by $\mu_a=10, p_{a,\mathcal{A}}=0.2,p_{a,\mathcal{B}}=0.35$ and $\mu_b=15, p_{b,\mathcal{A}}=0.3,  p_{b,\mathcal{B}}=0.4$, under high load ($\lambda=15$) and low load ($\lambda=10$). The optimal  threshold values  are marked using vertical lines. } }\label{fig4}
		 %\vspace*{-10pt}}
	\end{center}
\end{figure}

%%% TODO we should use \mathcal{G} and \mathcal{B} rather than A and B here

Figure~\ref{fig4}  illustrates  the multi-threshold policies obtained when solving the SMDP model under  deterministic and uniform transmission times with Gilbert-Elliott channels. 
 %We consider the same support for both channels. 
 %The results are seen in Figure~\ref{fig4}.
  The two channel states are denoted as $\mathcal{A}$ and $\mathcal{B}$. Each transmission time distribution  corresponds to \textit{two value functions}, for channels at states $\mathcal{A}$ and $\mathcal{B}$. Hence, each value function implies its own threshold. Observe that the value functions for $\mathcal{A}$ and $\mathcal{B}$ are very close to each other. Nonetheless,  the thresholds can be easily distinguished.  Under the Gilbert-Elliott channel model, for all the scenarios considered we were always able to find a separate threshold for each channel type.
  %This observation stems from the entire set packet loss values and probably from other parameters. 
%  The threshold structure is observed in all simulated distributions, even though the proof only covers the exponential case. See that 
 % 
  While a rigorous analysis of the multi-threshold policy is subject for a future work, the numerical analysis presented here can be used to devise heuristics to be concurrently applied with value iteration, aiming towards a faster convergence.

\subsection{Impact of thresholds}

Next, we illustrate the impact of the thresholds on the throughput.   Our goals are to assess, 1) for a given
service distribution, how
the throughput varies as a function of the threshold, and 2) how the service distribution impacts the optimal
threshold. To this aim, we consider exponential, deterministic and uniform service distributions.

Policy evaluation was performed using three different techniques under the exponential, deterministic and uniform service time distributions.  
Under exponentially distributed service times,  given a fixed threshold policy the resulting system dynamics is governed by a  continuous-time Markov chain (CTMC). Impulse rewards are accumulated at service completions, and system throughput 
is given by the  expected 
 impulse reward   of the resulting CTMC. 
For the deterministic case, we rely on a special class of  solution methods  to efficiently compute metrics of interest for models wherein  all events are exponentially distributed except for a single deterministic one~\cite{edmundo}.   For uniform service times, we make use of~\cite{ballarini2013transient}.  
The results for the exponential and deterministic service times were obtained using the Tangram~II tool~\cite{e2000tangram}, whereas the uniform case was evaluated using  the Oris tool~\cite{ballarini2013transient}. More details are found in Appendix~\ref{appthr}.

In Figure~\ref{thr} we let $B=10$ and allow the threshold to vary between 0 and 9.  
We let $\lambda=17$, $ C=1$, $p_b=0.42$, $p_a=0.25$, $\mu_b=13$ and $\mu_a=10$.
 Note that 
as the threshold increases
the throughput
first increases and then decreases.  A threshold of 0 (resp., 9) consists of always transmitting 
through  the less reliable channel, i.e., channel ``b'' (resp., through the most reliable channel, i.e.,  channel ``a''). The optimal threshold equals 3, 4 and 6 for deterministic, 
uniform and exponential service distributions.  As the variability in the service time increases,
 the optimal strategy privileges transmissions through the most reliable  channel.

\begin{figure}[h!]
		\includegraphics[angle=0, width=0.50 \textwidth]{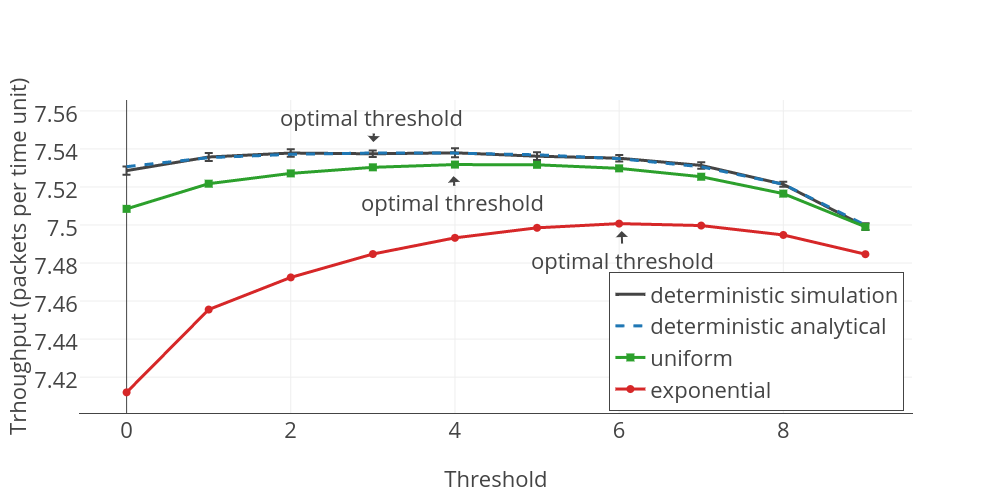}%
	\caption{Throughput as a function of threshold $(B=10)$.  
	As  the variability in the service time distribution increases, the  optimal threshold increases.}\label{thr}
\end{figure}

\begin{figure}[h!]
		\includegraphics[angle=0, width=0.50 \textwidth]{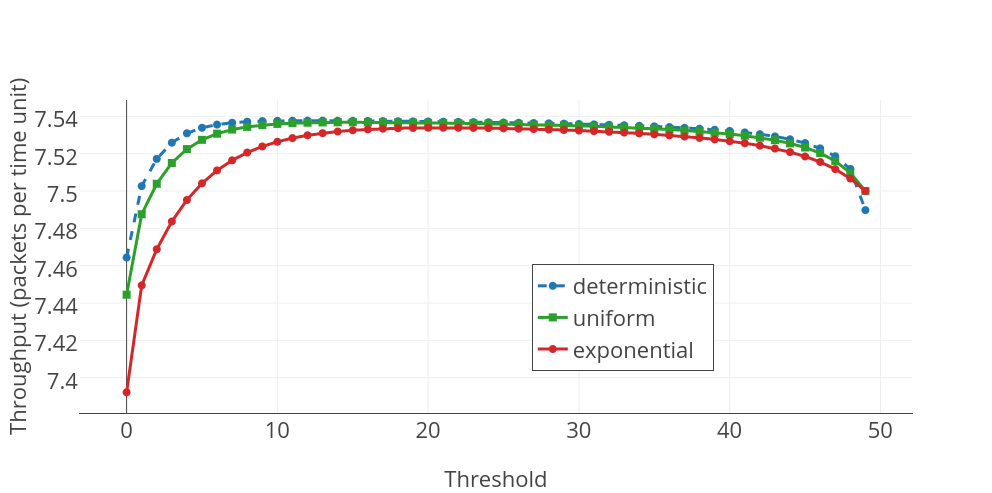}%		
	\caption{Throughput as a function of threshold  $(B=50)$.  As the threshold increases, the throughput first increases and then decreases. }\label{thr50c}
\end{figure}

Next, we let $B=50$ and $\lambda=13$, and keep the other parameters unchanged (Figure~\ref{thr50c}).
Similar observations 
as made in the previous paragraph apply.  In all the considered cases, there is a unique optimal threshold.
  The optimal threshold equals 12, 15 and 21 for  deterministic, 
uniform and exponential service distributions (see Figure~\ref{thr50a}).

As a sanity check, we also ran simulations for   deterministic service times, using the Tangram II tool~\cite{e2000tangram}.  
For each set of parameters, we ran 30  simulation runs. Each run lasted for 100.000  time units. We confirmed 
that the results obtained through simulations and analytically are in conformance.  The  95\% 
confidence intervals in
 Figures~\ref{thr50b} 
and~\ref{thr}, obtained through simulations, are in agreement with the analytical results, reported 
as dotted lines.

\begin{figure}[h!]
		\vspace{-0.1in}
		\includegraphics[angle=0, width=0.50 \textwidth]{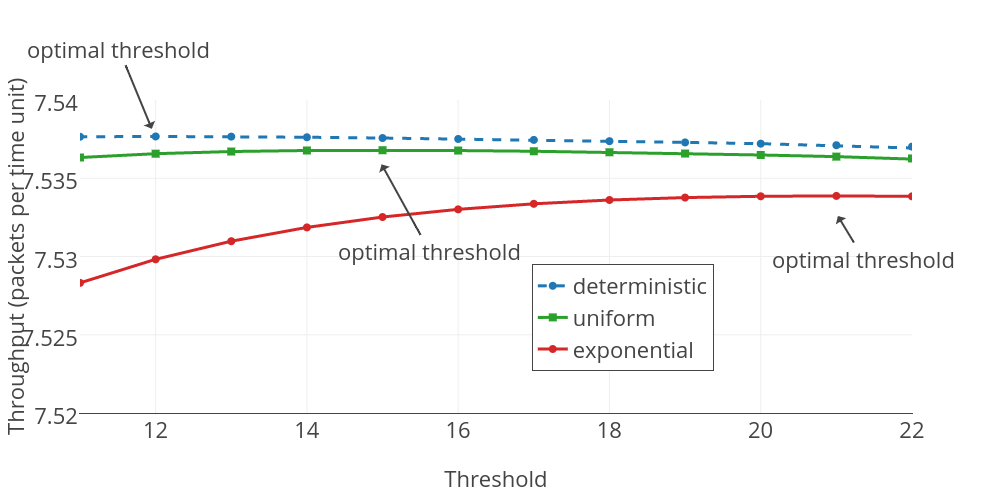}%		
	\caption{Throughput as a function of threshold  $(B=50)$ (zoom of Figure~\ref{thr50c}). As  the variability in the service time distribution increases, the  optimal threshold increases.}\label{thr50a}
\end{figure}

\begin{figure}[h!]
		\vspace{-0.1in}
		\includegraphics[angle=0, width=0.50 \textwidth]{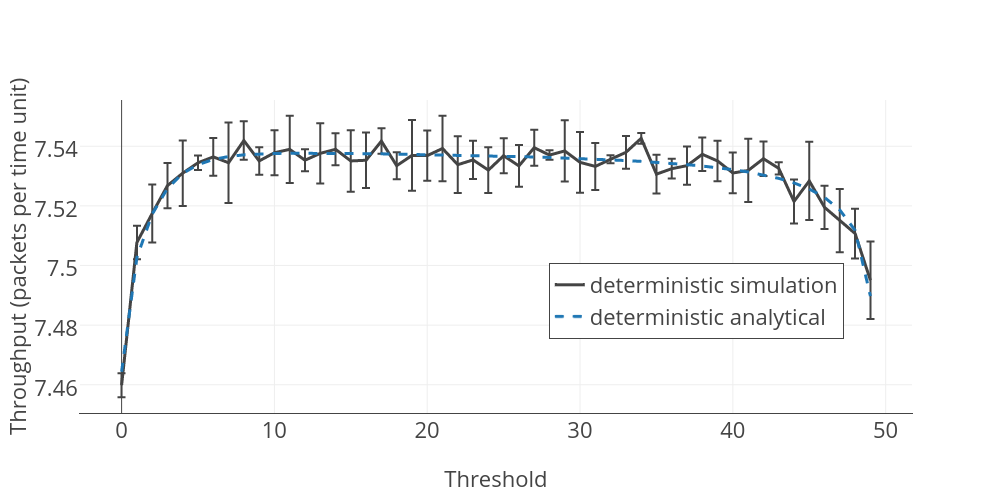}%		
	\caption{Throughput as a function of threshold  $(B=50)$. Results obtained analytically and through simulations are in agreement.}\label{thr50b}
\end{figure}

As indicated above, the numerical examples suggest that the optimal threshold increases with respect
to  the variability in the service time distribution.  Verifying such a conjecture is subject for  future work.
Note also that in the illustrative results reported in this  section the  throughput 
values vary between 7.4 and 7.6.  % in a short range.
  Although in the presented examples the range of throughput values is short, in general it  can be arbitrarily large, further motivating the search for the optimal threshold 
  value.

% Average transmission rates and packet losses are given by

\vspace{-0.05in}

\section{Conclusion}

We have proposed an SMDP model for optimal PHY configuration,   derived equations for the value function  for several interesting cases, and formally    shown structural properties of the optimal policy  when transmission times are exponentially distributed.  In particular, we have shown the existence of optimal policies of threshold type. The numerical solution of the proposed model indicates the existence of multi-threshold policies for Gilbert-Elliott channels.  Showing the optimality of the latter under  general settings is an interesting open problem.     

\vspace{-0.05in}

\ifnconf
\appendix
%\chapter{A}
\subsection{Bellman equation for the exponential case}\label{app:pBE}

Next, our goal is to derive the Bellman equations~\eqref{eq:lp104}-\eqref{eq:bellman2} from the process model.  We consider exponentially distributed service times.

Denote by $u_n=0$ the decision to send through channel ``b" and $u_n=1$ through channel ``a".
\begin{comment}
$\tilde{D}_b$ and $\tilde{D}_a$ count successful transmissions on direct and relay channels.
Note that $d\tilde{D}_b(s)$ (resp., $d\tilde{D}_a(s)$) equal to $1$ if  a successful packet reception occurs 
at time $s$, through the direct channel (resp., relay), and zero otherwise.  
\end{comment}
We assume a stationary policy. % $u_n$ is stationary.

%Denote the rewards  received as the transmission starts as $c_a$ or $c_b$.
%See that $R(t)=0$ implies $dD_a(s)=0$.
Note that the reward ($c_a$ or $c_b$) is immediately received at $V_n(t)$.
Take infinitesimal $\theta$, such that probability that more than one event occur  in
 $[0,\theta]$ goes to zero. 
\begin{comment} 
  Let $dD_a(s)$ and $dD_b(s)$ 
be the  time-shifted processes corresponding to $d\tilde{D}_a(s)$ and $d\tilde{D}_b(s)$, respectively.
Let $\tau_a$ and $\tau_b$  be random variables characterizing the time it takes for a transmission through the relay and directly, respectively. 
\end{comment}
 
Use dynamic programming principle to write 
\begin{equation}
 V_n(t)=\int_t^{t+\theta} e^{-\gamma s} r(s)dR(s) +e^{-\gamma(\theta+t)}V_n(t+\theta), 
\end{equation}
where $r(s)$ denotes reward at time instant $s$, $R(s)$ is Poisson counting process which mean will be specified below.  
Due to the stationary property we will assume the initial time $t=0$, and will omit the time mark where it is clear. 
To assume the immediate reward we take $\theta$ infinitesimal. Denote $\int_0^\theta e^{-\gamma s}r(s)dR(s)=r_\theta$.

Then, write
\begin{align*}
 & V_n= r_\theta+ \\
 & \qquad +e^{-\gamma\theta}\Big[\theta\lambda(1-u_n)V^{b,A}_n+\theta\lambda u_nV^{a,A}_n\\
&\qquad \theta \mu_au_nV_{n-1}+\theta\mu_b(1-u_n)V_{n-1}+\\
& \qquad V_n(1-\theta\lambda(1-u_n)-\theta\lambda u_n-\theta\mu_au_n-\theta\mu_b(1-u_n))\Big]
\end{align*}
%where the first term stands for the  product of reward the average transmission rate, the routing decision and the channel state at transmission time. 
The term in brackets sums up all possible outcomes for the value function at time $\theta$ weighted by the corresponding probabilities. (e.g. arrivals happen w.p. $\lambda\theta$).
Now, since $\theta$ is small, substitute $e^{-\gamma\theta}\approx1-\gamma\theta$, % multiply both sides by $\gamma$
\begin{align*}
& V_n= r_\theta+(1-\gamma\theta)[\theta\lambda(1-u_n)V^{b,A}_n+\theta\lambda u_nV^{a,A}_n+\\
&\qquad \theta\mu_au_n V_{n-1}+\theta\mu_b(1-u_n)V_{n-1}+\\
& \qquad V_n(1-\theta(1-u_n+u_n)\lambda-\theta\mu_au_n-\theta\mu_b(1-u_n))]
\end{align*}
See that $\theta^2\to0$
\begin{align*}
& V_n= r_\theta+[\theta\lambda V_{n+1}+\\
&\qquad \theta\mu_au_n V_{n-1}+\theta\mu_b(1-u_n)V_{n-1}+\\
& V_n(1-\theta\lambda-\theta\mu_au_n-\theta\mu_b(1-u_n))-\gamma\theta V_n]
\end{align*}
Denote $\delta^{-1}=\lambda+\mu_bu+\mu_a(1-u_n)+\gamma$. See that $V_n$ cancel on both sides. %and divide by $\gamma\theta$
\begin{align*}
& 0= r_\theta+[\lambda(1-u_n) V^{b,A}_n+\lambda(u_n) V^{a,A}_n+\\
&\qquad \mu_au_n V_{n-1}+\mu_b(1-u_n)V_{n-1}
-\delta^{-1} V_n]\theta
\end{align*}
Calculate now $r_\theta$.
The reward is received in all possible cases. That is, whenever the process does not change, arrival happens or transmission ends. Hence, at $V_n$ we have a reward approximately accumulated with Poisson process with mean rate $\lambda+\mu_au_n+\mu_b(1-u_n)+\gamma$, where the terms stand for summation of the rates of arrival, transmission on path $a$, transmission on path $b$ and the discounting rate.  
Hence, $dR=dt(\lambda+\mu_au_n+\mu_b(1-u_n)+\gamma)$.  Substitute %Multiply and divide $r_\theta$ by $dt$,
\begin{align*}
&\int_0^\theta e^{-\gamma t}r(s)dR(s)=\\
& =(\lambda+\mu_au_n+\mu_b(1-u_n)+\gamma)C\int^\theta_0e^{-\gamma t}dt \\ % + o(\theta) \\
&=(\lambda+\mu_au_n+\mu_b(1-u_n)+\gamma)C\frac{1-e^{\gamma\theta}}{\gamma}  \\ %+ o(\theta)  \\
& = \theta(\lambda+\mu_au_n+\mu_b(1-u_n)+\gamma)C % + o(\theta),
\end{align*}
where  $r(s)=C$ %+ o(\theta)$ 
for $s \in (0, \theta)$, while  $C=(1-u_n)c_b+u_nc_a$.  Note that 
%$c_b=(1-p_b)\frac{\mu_b}{\gamma+\mu_b}$,
\begin{equation*}
c_b=(1-p_b)\int_0^\infty e^{-\mu_b t} e^{-\gamma t} dt
=(1-p_b)\frac{\mu_b}{\gamma+\mu_b}
\end{equation*}
$c_b$ stands for average reward equal to $1-p_b$, discounted according to 
the transmission time. $c_a$ is equivalently calculated. 
Finally, substitute the reward
\begin{align}
& V_n\theta\delta^{-1}=\theta(\lambda+\mu_au_n+\mu_b(1-u_n)+\gamma)C+\\
&[\lambda(1-u_n) V^{b,A}_n+\lambda(u_n) V^{a,A}_n+ \nonumber \\
&\qquad \mu_au_n V_{n-1}+\mu_b(1-u_n)V_{n-1}
]\theta + o^2(\theta)
\end{align}
Since the routing decision is applied on $u_n$, the Bellman equation follows by the maximization over $u_n=\{0,1\}$. Substitute the $\max$ operator, divide by $\theta\delta^{-1}$, and let $\theta \rightarrow 0$,
\[
V_n=\max[c_a+\mu_a\delta_aV_{n-1}+\lambda\delta_aV^{a,A}_n,c_b+\mu_b\delta_bV_{n-1}+\lambda\delta_bV^{b,A}_n]
\]
\qed

\fi
%\appendix
%\chapter{B}
\ifnconf
%\section{Threshold policy for the exponential case}

\subsection{Proof of lemma~\ref{mainlemma1}}\label{app:pr_lem}
For simplicity we assume all packets are equally rewarded by $r=1$.
Denote $c_b=(1-p_b) \beta_b$ and $c_a=(1-p_a) \beta_a$.

%W will use the following lemma 

\begin{proof} %[Proof of lemma~\ref{mainlemma1}]

To show that operators $\A_b,\A_a,\T$ preserve properties 1)-4) construct first some \rdo{\{U, U^{b,A}, U^{a,A}\} }, prior to applying the operators on them, such that $U^{b,A}$, $U^{a,A}$ and $U$ are non-decreasing, concave and bounded. Moreover, the further construction of $U^{b,D}$ and $U^{a,D}$ using \rdo{\{U, U^{b,A}, U^{a,A}\}} and applying~\eqref{eq:lp113} on them, is such that $U^{b,D}_n-U^{a,D}_n$ is non-decreasing in $n$. Note that by non-decreasing property of $U^{b,A}$, $U^{a,A}$ and $U$ such a construction is straight-forward (we omit the detailed construction procedures). 
Also note that by concavity of $U^{b,A}_n$, $U^{a,A}_n$ and $U_n$, $U^{b,D}_n$ and $U^{a,D}_n$ are also concave in $n$, since  they are positive linear sum of concave functions.  
Hence, \rdo{\{U, U^{b,A}, U^{a,A},U^{b,D},U^{a,D}\}\in \calS }.
 %Due to the space limitation, we skip the proof for the boundary conditions.

%Observe from~\eqref{eq:lp113} and~\eqref{eq:lp104} that following relation holds
%\begin{align}
%& U^{{b,D}}_{n}=U^{{b,A}}_{n-1}+c_b   \label{eq:lp105} \\
%& U^{{a,D}}_{n}=U^{{a,A}}_{n-1}+c_a \label{eq:lp106}
%\end{align}
Denote  $D_n=\mu {\it _b}\delta {
\it _b}U_{{n}}+\lambda\delta_b{\it U^{{b,A}}}_{{n+1}
}=$ and $R_n=\mu {\it _a}\delta {
\it _a}U_{{n}}+\lambda\delta_a{\it U^{{a,A}}}_{{n+1}
}$. By construction, $D_n=U^{b,D}_n-c_b$, $R_n=U^{a,D}_n-c_a$, hence, are concave.
By property $1$, $D_n-R_n$ is an increasing sequence.
%In addition, there is at most one threshold state $n=t$ such that  $U^{b,A}_m<U^{a,A}_m,\;m<t$ and $U^{b,A}_m>U^{a,A}_m,\;m\geq t$.

At $n=0$, $D_0=U^{b,D}_0$, $R_0=U^{a,D}_0$.
By boundary condition~\eqref{eq:pth14}, $D_0=R_0$, hence it always holds $D_n>R_n$. 
To this end, define the slope of some discrete $W_n$, $n\in\{0\cdots B-1\}$ as $\Delta(W_n)=\Delta W_n=W_{n+1}-W_n$.
Clearly, 
\begin{equation}
\Delta(D_n)>\Delta(R_n)\label{ineq:1}
\end{equation}
By concavity $\Delta(D_n)$ and $\Delta(R_n)$ are decreasing sequences, hence $\Delta(D_n)+\Delta(R_n)$ is also decreasing.
Next, write
\[
\Delta(D_n)+\Delta(R_n)=(D_n-R_{n-1})+(R_n-D_{n-1})
\]
Clearly, at least one of the two sequences $D_n-R_{n-1}$ and $R_n-D_{n-1}$ is decreasing.
Assume that, $D_n-R_{n-1}$ is increasing. However, then $\Delta(R_n-D_{n-1})>\Delta(D_n-R_{n-1})$, hence $\Delta(R_n)+\Delta(R_{n-1})>\Delta(D_n)+\Delta(D_{n-1})$, which contradicts~\eqref{ineq:1}, thus the assumption above is incorrect and the sequence $D_n-R_{n-1}$ is increasing in $n$. 

We show next that applying the corresponding operators results in functions which possess these properties as well. Namely, \rdo{\{\T U, \A_bU^{b,A}, \A_aU^{a,A},\T_bU^{b,D},\T_aU^{a,D}\}\in \calS }.
Note that the preservation is separately proved for the general state and for the boundary conditions.

\subsubsection*{{{\textbf{Property 1}}} [${U^{b,D}_{n}-U^{a,D}_{n}\text{  is increasing in n}}$]}

%\shifrinA{my version, pay attention, no capital letters.}

% \shifrin{\it \textbf{property~\eqref{eq:pth1}}.}
In order to prove that
\req{0}{U^{b,D}_{n}-U^{a,D}_{n}\text{  is increasing in n}}{eq:lp101}

Write~\eqref{eq:lp101} as follows,
\req{0}{
\mu {\it _b}\delta_bU_{{n-1}}-\mu {\it _a}
\delta_aU_{n-1}+\lambda\delta_b{\it U^{b,A}}_{{n}}-\lambda\delta_a{\it U^{a,A}}_{n}+{\it c_b}-{
\it c_a}
}{eq:lp102}
Note that~\eqref{eq:lp101} is equivalent to $D_n-R_n$.
%We show that the triple of operators $\A_{d},\A_{r},\T$ jointly preserves  property~\eqref{eq:pth1}.
% Substitute~\eqref{eq:pth6} and ~\eqref{eq:pth7} to apply the operators to~\eqref{eq:pth1}
Apply the operators to~\eqref{eq:lp102},
\rals{-1}{
&\Theta_n=\mu {\it _b}\delta_b\T U_{{n-1}}-\mu {\it _a}
\delta_a\T U_{{n-1}}+\\
&\lambda\delta_b\A_b{\it U^{b,A}
}_{n}-\lambda\delta_a\A_a{\it U^{a,A}}_{n}+
{\it c_b}-{\it c_a}
}
%Our goal is to show that $\tilde{\Delta}_n \ge 0$.
Divide the proof into two possible cases cases, depending on whether
$\T U_{{n-1}}=\T_b U_{n-1}$ or $\T U_{{n-1}}=\T_a U_{n-1}$

\begin{enumerate}
\item \underline{$\T U_{{n-1}}=\T_b U_{n-1}$}.
Then, %and omit the constants $c_a$ and $c_b$
\begin{align}
&\Theta_n=\mu {\it _b}\delta_b \left( \mu {\it _b}\delta {
\it _b}U_{{n-2}}+\lambda\delta_b{\it U^{b,A}}_{{n-1}
}+c_b \right) \nonumber \\
&-\mu {\it _a}\delta_a \left(
\mu {\it _b}\delta_bU_{{n-2}}+\lambda\delta {\it
_b}{\it U^{b,A}}_{n-1}+c_b \right) \nonumber\\
& +\lambda\delta {
\it _b} \left( \mu {\it _b}\delta_bU_{{n}}+
\lambda\delta_b{\it U^{b,A}}_{{n+1}} \right)  \nonumber \\
& -\lambda
\delta_a \left( \mu {\it _a}\delta_aU_{{n
}}+\lambda\delta_a{\it U^{a,A}}_{n+1}\right)+c_b-c_a \label{eq:pth2}
\end{align}

Write
\begin{align}
&{\Theta}_n= \mu {\it _b}\delta_b(D_{n-2}+c_b ) -\mu {\it _a}\delta_a( D_{n-2}+c_b ) \nonumber \\
&\lambda\delta {
\it _b}\left( D_n \right) -\lambda
\delta_a \left( R_n\right)+c_b-c_a \nonumber \\
&=
 \mu {\it _b}\delta_b(D_{n-2}+c_b ) -\mu {\it _a}\delta_a( D_{n-2}+c_b ) \nonumber \\
 &+\lambda\delta {
\it _b}\left( D_n \right)  -
\lambda
\delta_a \left( R_n \right)+c_b-c_a \nonumber\\
& - \lambda\delta_a ( D_n)+\lambda\delta_a ( D_n)\label{eq:pth3}
\end{align}

See that $D_n$ and $R_n$ are concave due to the concavity of $U^{d,a}_{n}$ and $U^{r,a}_{n}$. Now observe that 
\[
\mu_b\delta_b-\mu_a\delta_a+\lambda\delta_b-\lambda\delta_a=\gamma\delta_b\delta_a(\mu_b-\mu_a)>0
\]
Use the bound for the slope twice, that is
\(
U_{n}<U_{n-1}+K\text{  and  } U_{n}<U_{n+1}+K
\) and apply it to $D_n$. Observe that $\lambda\delta_b-\lambda\delta_a<0$, and denote $D_n=D_{n-2}+\phi_1(n)$.
Hence, 
$(\lambda\delta_b-\lambda\delta_a)D_n=(D_{n-2}+\phi_1(n))(\lambda\delta_b-\lambda\delta_a)$. Denote $\phi_2(n)=\phi_1(n)(\lambda\delta_a-\lambda\delta_b)$ and write
%and omit again the constants. Write
\rals{0}{
&{\Theta}_n = (\mu_b\delta_b-\mu_a\delta_a+\lambda\delta_b-\lambda\delta_a)D_{n-2}\\
&+\lambda\delta_a ( D_n-R_n)-\phi_2(n)+\kappa,
}
where $\phi(n)$ is positive decreasing by concavity of $D(n)$ and $\kappa$ is a suitable constant.
Since the expression above is a linear combination of increasing functions with positive coefficients and constants, it is increasing. Hence, the operators do preserve property 1) in this case.

\item \underline{$\T U_{{n-1}}=\T_a U_{n-1}$}. 
The proof is analogous to that of the first case.
%Write
%\begin{align}
%\tilde{\Delta}_n=& \mu {\it _b}\delta_b(R_{n-2}+c_a ) -\mu {\it _a}\delta_a( R_{n-2}+c_a ) \nonumber +\lambda\delta {
%\it _b}\left( D_n\right) -\lambda
%\delta_a \left( R_n\right)+c_b-c_a
%\end{align}
Similarly to~\eqref{eq:pth3} write
\begin{align*}
{\Theta}_n=& \mu {\it _b}\delta_b(R_{n-2}+c_a ) -\mu {\it _a}\delta_a( R_{n-2}+c_a ) \nonumber +\lambda\delta {
\it _b}\left( D_n\right) \\
&-\lambda
\delta_a \left( R_n\right)+c_b-c_a-\lambda
\delta_b \left( R_n \right)+\lambda
\delta_b \left( R_n\right)
\end{align*}
Again, use the property of the bounded slope and write
\rals{0}{
&{\Theta}_n=(\mu_b\delta_b-\mu_a\delta_a+\lambda\delta_b-\lambda\delta_a)R_{n-2}\\
&+\lambda\delta_b(D_n-R_n)-\phi_3(n)+\kappa_1,
}
where $\phi_3(n)$ is positive decreasing by concavity of $D(n)$ and for some suitable $\kappa_1$. Since the expression above is a linear combination of increasing functions with positive coefficients and constants it is increasing. Hence, the operators preserve property 4) in this case as well.

%To prove for the boundary condition at $0$ substitute $n=1$
%\[
%\mu_b\delta_b\T U_{0}-\mu {\it _a}
%\delta_a\T U_{0}+\lambda\delta_b\A_b{\it U^{b,A}
%}_{,1}-\lambda\delta_a\A_a{\it U^{a,A}}_{,1}+{\it c_b}-{
%\it c_a}
%\]

\ifnconf
\else
\item \underline{\emph{Boundary conditions}}

Note that at $n=0$ there is no decision to be made, while the proof for $n=1$ is identical to that of $1<n<\BS$.
At state $\BS$ observe that 
\req{0}{\A_b U^{d,a}_{\BS}=\A U^{d,a}_{\BS-1}\text{ and }\A_a U^{r,a}_{\BS}=\A U^{r,a}_{\BS-1}}{eq:lp107}. %Hence
%Next, we show that  property 4 holds at state $\BS$. 
%Assume $r$ was the decision at $\BS-1$ while $d$ was the decision at $\BS$, apply the operators and 
Write 
\begin{align*}
&\mu {\it _b}\delta_b\T U_{{\BS-1}}-\mu {\it _a}
\delta_a\T U_{{\BS-1}}+\\
&\lambda\delta_b\A_b{\it U^{b,A}
}_{\BS}-\lambda\delta_a\A_a{\it U^{a,A}}_{\BS}+
{\it c_b}-{\it c_a}-\\
&\big(\mu {\it _b}\delta_b\T U_{{\BS-2}}-\mu {\it _a}
\delta_a\T U_{{\BS-2}}+\\
&\lambda\delta_b\A_b{\it U^{b,A}
}_{\BS-1}-\lambda\delta_a\A_a{\it U^{a,A}}_{\BS-1}+
{\it c_b}-{\it c_a} \big)=\\
& \mu {\it _b}\delta_b\T U_{{\BS-1}}-\mu {\it _a}
\delta_a\T U_{{\BS-1}}-\\
&\big(\mu {\it _b}\delta_b\T U_{{\BS-2}}-\mu {\it _a}
\delta_a\T U_{{\BS-2}}\big)\geq0\\
%&(\mu_b\delta_b-\mu_a\delta_a)(T U_{{\BS-1}}-T U_{{\BS-2}})=
%& U^{b,D}(\BS)-U^{a,D}(\BS)-\left( U^{b,D}(\BS-1)-U^{a,D}(\BS-1) \right) \\
%= &\mu {\it _b}\delta_bU_{\BS-1}-\mu {\it _a}
%\delta_aU_{\BS-1}+\lambda\delta_b{\it U^{b,A}
%}_{{\BS}}-\lambda\delta_a{\it U^{a,A}}_{\BS}+\\{\it c_b}-
%&{\it c_a} -\big(\mu {\it _b}\delta_bU_{\BS-2}-\mu {\it _a}
%\delta_aU_{\BS-2}+\lambda\delta_b{\it U^{b,A}
%}_{{\BS-1}}- \\
%&\lambda\delta_a{\it U^{a,A}}_{\BS-1}+
%{\it c_b}-{\it c_a}\big)
\end{align*}
%We want to show that $\tilde{\Delta}_Q \ge 0$, where
%\[
%\tilde{\Delta}_Q = \mu {\it _b}\delta_b\T U_{\BS-1}-\mu {\it _a}
%\delta_a\T U_{\BS-1}-\left(\mu {\it _b}\delta_b\T U_{\BS-2}-\mu {\it _a}
%\delta_a\T U_{\BS-2}\right)
%\]
The inequality above is true by the increasing property of $U_n$. 
%After a few simple algebraic manipulations, it is possible to verify that the inequality
%$\tilde{\Delta}_Q \ge 0$  holds  
\fi
\end{enumerate}

\subsubsection*{\textbf{Property 2} [Concavity]}

%\subsection{Concavity}\label{app:Conc}
%\begin{lemma}[Concavity]
%$V^{b,A}_{n}$ resp. $V^{a,A}_{n}$ are concave for all $n\geq1$, resp. $n\geq0$.
%\label{lemC}
%\end{lemma}

Apply corresponding operators, $u\in\{a,b\}$:
\[
\A_uU^{(u,A)}_{n}=\mu_u\delta_uU_{n}+\lambda\delta_uU^{u,A}_{n+1} 
\]
The result is concave due to the concavity of $U_{n},U^{a,A},U^{b,A}$.

The result for $U^{u,D}$ is similarly deduced.
%We show that $U_{n+1}-U_n$ is decreasing.

Apply $\T$ on $U_n-U_{n-1}\geq U_{n+1}-U_n$.
We show only non-trivial cases where the threshold is within the range $[n-1,n]$. In the case $n=t$, $D_{t-1}+c_b<R_{t-1}+c_a$
\[
D_n-R_{n-1}+c_b-c_a\geq D_{n+1}-R_{n}+c_b-c_a\geq D_{n+1}-D_n,
\]
where the first inequality follows from the fact the sequence $D_{n+1}-R_n$ is decreasing.
Apply $\T$ on $U_n-U_{n-1}\geq U_{n+1}-U_n$ in the case $n=t+1$
\[ 
R_n-R_{n-1}>D_n+c_b-c_a-R_{n-1}\geq D_{n+1}+c_b-R_n-c_a,
\]
where the second inequality is due to the fact the sequence $D_{n+1}-R_n$ is decreasing.

\begin{comment}

\begin{align*}
& V_0^{b,A}=V_0^{a,A}=V_1^{a,D}=\\
& V_0^{a,D}\mu_a\delta_a+V_1^{a,A}\lambda\delta_a+c_a
\end{align*}
\end{comment}
%\input{Conc_lem_proof}
%$U^{(b,A)}$ and $U^{(b,A)}$ are concave, hence $U^{(a,D)}$ and $U^{(b,D)}$ are concave as well.

\subsubsection*{\textbf{Property 3} [Non-decreasing property]}
%\subsection*{\{{ \textbf{}}}}[]

We show that $\T U_n, \A_b U^{(b,A)}_n, \A_a U^{(a,A)}_{n}$, $n\in\{1,\cdots,\BS\}$ are non-decreasing.

% The proof immediately follows from applying the operators.

 For $\T U_n$,  we assume that $U_n \ge U_{n-1}$ and we show that $\T U_n \ge \T U_{n-1}$. Indeed,
\[
\mu_b\delta_bU_{n-1}+\lambda\delta_bU^{(b,A)}_n+c_b\geq\mu_b\delta_bU_{n-2}+\lambda\delta_bU^{(b,A)}_{n-1}+c_b
\]
For $\A_b U^{(b,A)}_n$ we have,
\[
\mu_b\delta_bU_{n}+\lambda\delta_bU^{(b,A)}_{n+1}\geq\mu_b\delta_bU_{n}+\lambda\delta_bU^{(b,A)}_{n+1},
\]
with equality at $\BS$. Similarly, it can be shown that  the result  holds for $\A_a$.

%\ifnconf
\subsubsection*{{{ \textbf{Property 4}}}[Slope bound]}

 %To show that the slope is bounded we first prove that $\T U_{n}-\T U_{n-1}\leq K$, then
%$ \A_b U^{b,A}_n- \A_b U^{b,A}_{n-1}\leq K$ and finally, $\A_a U^{a,A}_{n}- \A_a U^{a,A}_{n-1}\leq K$.

\begin{enumerate}
\item \underline{\emph{$\T U_{n}-\T U_{n-1}\leq K$¸}}

%\[
%\T U_{n}-\T U_{n-1}\leq K
%\]
Note that $\lambda\delta_b+\mu_b\delta_b<1$.
We divide the proof into four cases.

\underline{\emph{Case 1: } $\T U_{n} = \T_b U_{n}$ and  $\T U_{n-1} = \T_b U_{n-1}$}.
In this case,
\begin{align*}
& \T U_{n}-\T U_{n-1} \leq \\
&\leq  \mu_b\delta_bU_{n-1}+\lambda\delta_bU^{b,A}_n+c_b-\mu_b\delta_bU_{n-2}-\lambda\delta_bU^{b,A}_{n-1}-c_b \\
& \leq \mu_b\delta_bK+\lambda\delta_bK, \\
& \leq K
\end{align*}
Hence, the property holds.

\underline{\emph{Case 2: } $\T U_{n} = \T_a U_{n}$ and  $\T U_{n-1} = \T_a U_{n-1}$}.
The proof is identical to case 1, and is omitted.

\underline{\emph{Case 3: } $\T U_{n} = \T_b U_{n}$ and  $\T U_{n-1} = \T_a U_{n-1}$}.
% Next, assume that "d" is chosen for $n$ and "r" is chosen for $n-1$.
Write
\begin{align*}
& \T U_{n}-\T U_{n-1} \leq \\
&\mu_b\delta_bU_{n-1}+\lambda\delta_bU^{b,A}_n+c_b-\mu_a\delta_aU_{n-2}-\lambda\delta_aU^{a,A}_{n-1}-c_a\leq\\
&\mu_b\delta_bU_{n-1}+\lambda\delta_bU^{b,A}_n+c_b-\mu_b\delta_bU_{n-2}-\lambda\delta_bU^{b,A}_{n-1}-c_b\leq \\
&\mu_b\delta_bK+\lambda\delta_bK \\
&\leq K
\end{align*}
Hence, the property holds.
\underline{\emph{Case 4: } $\T U_{n} = \T_a U_{n}$ and  $\T U_{n-1} = \T_b U_{n-1}$}.
This case contradicts property 4 and thus is excluded.

\item  \underline{\emph{$\A_b U^{b,A}_n- \A_b U^{b,A}_{n-1}\leq K$}}
To show the bound for $U^{b,A}$ write
\begin{align*}
&\A_bU^{b,A}_n-\A_bU^{b,A}_{n-1}=\\
&=\delta_b\mu_bU_{n}+\delta_b\lambda U^{b,A}_{n+1}-\delta_b\mu_bU_{n-1}-\delta_b\lambda U^{b,A}_n\\
&\leq\delta_b(\mu_b+\lambda)K\\
&\leq K
\end{align*}

\item \underline{\emph{$\A_a U^{a,A}_{n}- \A_a U^{a,A}_{n-1}\leq K$}}.
The proof  is identical to the previous one.

\item \underline{\emph{Boundary conditions}  }

Write boundary conditions for the state $\BS$. As $\A_bU_{\BS}^{b,A}-\A_bU_{\BS-1}^{b,A}=0$, we have  %Choose "d", write
\[
\A_bU^{b,A}_{\BS}-\A_bU^{d,d}_{\BS-1}<K
\]
The proof for $\A_a$ is similar. The prove for $\T$ at $\BS$ is similar to that at $n<\BS$. To see that write:
\begin{align*}
&\T_bU_{\BS}-\T_bU_{\BS-1}=\delta_b\mu_bU_{\BS-1}+\delta_b\lambda U^{b,A}_{\BS}-\\
&\delta_b\mu_bU_{\BS-2}-\delta_b\lambda U^{b,A}_{\BS-1}
%&\delta_b\mu_bU_{\BS}+\delta_b\lambda U_{d,\BS}-\delta_b\mu_aU_{\BS-1}-\delta_b\lambda U_{r,\BS}\leq\delta_b(\mu_b)K\leq K
\end{align*}

\item \underline{\emph{Constant $K$}  }

It is left to select the constant $K$. We do this by proving the bound for the boundary conditions at $0$.

% Assume  "d" for both "max" operators and substitute~\eqref{eq:expMDP1}

% \daniel{please, check if the conditions stated for cases 1 and 2 are correct}

We show that 
\begin{equation}
 U_{1}-  U_{0} \leq K \Rightarrow  \T U_{1}-  \T U_{0} \leq K 
\end{equation}

\begin{enumerate}
\item \underline{\emph{Case 1:}  } Assume $\T U_{1} = \T_b U_{1}$ and  $\T U_{0} = \T_b U_{0}$.

From~\eqref{eq:pth14} we have
\reqs{0}{
V^{b,D}_{0}= \T_b U_0 
= \lambda \bar\delta V^{b,A}_{0}
}
% \daniel{ are the equalities above correct? In particular, the passage from \eqref{pass1} to \eqref{pass2} is not clear to me. }
Then,
\begin{align}
&\T U_{1}-\T U_{0}= \nonumber \\
&\T_bU_{1}-\T_b U_0= \nonumber\\
% &\T_bU_{d,1}-\A_bU_{d,0}= \label{pass6}\\
& \mu_b\delta_bU_{0}+\lambda\delta_bU^{b,A}_{1}+c_b-\lambda\bar\delta U^{b,A}_{0} = \label{pass3} \\
&\mu_b\delta_bU_{0}-\mu_b\lambda\delta_b\bar\delta U^{b,A}_{0}+c_b+\lambda\delta_b( U^{b,A}_{1}- U^{b,A}_{0}) \label{pass4} \\
&\leq\mu_b\delta_bU_{0}-\mu_b \lambda\delta_b\bar\delta U^{b,A}_{0}+c_b+\lambda\delta_b K \nonumber
 %\leq\mu_b\delta_bK+\lambda\delta_bK\leq K
\end{align}
% \daniel{how do you go from \eqref{pass5} to \eqref{pass6}?}
%
%
% \daniel{how do you go from 
\eqref{pass4} is obtained from \eqref{pass3} using the following
% $\lambda\delta_bU_{d,1}-\lambda\bar\delta U_{d,0} =-\mu_b\gamma\delta_b\bar\delta U_{d,1}+\lambda\delta_b( U_{d,1}- U_{d,0})$.
\rals{0}{
&\bar\delta-\delta_b =  1/(\gamma+\lambda)-1/(\gamma+\lambda+\mu_b)\\
&= (\gamma+\lambda+\mu_b-\gamma-\lambda)/((\gamma+\lambda)(\gamma+\lambda+\mu_b)) \\
&=  \mu_b/((\gamma+\lambda)(\gamma+\lambda+\mu_b)) \\
&= \mu_b \delta_b \bar\delta
}
Denote $K_1=\mu_b\delta_bU_{0}-\mu_b\lambda\delta_b\bar\delta U_{0}^{b,A}+c_b$.
% Note that in state $0$ we have $U_{d,0}=U_{r,0}$.
\item \underline{\emph{Case 2:}.  } Assume $\T U_{1} = \T_a U_{1}$ and  $\T U_{0} = \T_a U_{0}$.
\reqs{0}{
V^{a,D}_{0}= \T_a U_0 
= \lambda \bar\delta V^{a,A}_{0}
}
% \daniel{ I don't understand the passages above}
Similarly, %to~\eqref{pass5a}-\eqref{pass5b}, 
\begin{align*}
&\T U_{1}-\T U_{0}=\\
& \T_a U_{1}^{a,A}-\T_a U_{0}^{a,A} \\
&\leq \mu_a\delta_aU_{0}-\mu_a\lambda\delta_a\bar\delta U_{0}^{a,A}+c_a+\lambda\delta_aK 
 %\leq\mu_b\delta_bK+\lambda\delta_bK\leq K
\end{align*}
% \daniel{please, explain passages above}
Denote $K_2=\mu_a\delta_aU_{0}-\mu_a\lambda\delta_a\bar\delta U_{0}^{a,A}+c_a$.

\item \underline{\emph{Case 3:}  }Assume $\T U_{1} = \T_b U_{1}$ and  $\T U_{0} = \T_a U_{0}$.
This case is identical to case 1, because $\T_a U_{0}=\T_b U_{0}$.
\end{enumerate}
$K$ has to satisfy $K_1 + \lambda \delta_b K  \leq K$ and 
$K_2 + \lambda \delta_a K  \leq K$, that is, 
$K_1 \leq (\mu_b+\gamma)\delta_bK$ and $K_2 \leq (\mu_a+\gamma)\delta_aK$.
Setting $K=\max\left\{\frac{K_1}{(\mu_b+\gamma)\delta_b},\frac{K_2}{(\mu_a+\gamma)\delta_a}\right\}$ 
suffices to  bound the slope in all cases.
\end{enumerate}

%\fi

\vspace{0.1in}

This finishes the proof of lemma~\ref{mainlemma1}.

\end{proof} 
\fi
%\end{appendices}
\ifnconf
\subsection{Throughput assessment}

\label{appthr}

In this section we assume that a  threshold policy is given.   Our aim is to assess the throughput.  
Recall that given threshold $T$, the policy consists of transmitting through the most reliable but
slowest channel (channel ``a'',  which is, for instance, a relayed channel) if the 
number of packets in the buffer
is smaller than or equal to $T$ at the time at which the transmission starts, and transmitting
 through the less reliable but fastest channel (channel ``b'', also known as the  direct channel) otherwise.  The policy is summarized in Table~\ref{tab:thrass}. 
Note that we assume that a packet is removed from the buffer immediately after its transmission ends.  Decisions are made
immediately after 1) a packet arrives  to an empty system, or 2) a packet is  removed from the buffer which remains occupied.

\begin{table}[h!]
\center
\caption{Two considered paths}
\begin{tabular}{l||l|l|l}
\hline
channel & reliability & rate  & condition for usage \\
 & &  & (number of buffered packets is) \\
\hline
$a$ & $1-p_a$, higher & $\mu_a$, lower  & smaller than or equal to $T$ \\
$b$ & $1-p_b$, lower &  $\mu_b$, higher &  greater than  $T$ \\
\hline
\end{tabular} \label{tab:thrass}
\end{table}

\emph{Note that once  the SMDP policy is fixed and given, our object of study is a semi-Markov process (SMP).  %  Semi-Markov decision processes (SMDPs) are based on SMPs.
 In what follows, we show how to compute the expected impulse reward of the SMP of interest.  The resulting metric is the throughput.  }
 
%  In this appendix, we refer to channels $a$ and $b$ as direct and relay channels, respectively.  
Assume that we are given matrix $\bP^{(a)}$ (resp., $\bP^{(b)}$), whose $(i,j)$ entry  indicates  the probability 
that the system contains $j$ packets after a transmission through channel $a$ (resp., $b$), given that it contained $i$ packets before
the transmission.  $\bP^{(a)}$ and $\bP^{(b)}$ are the transition matrices of the embedded chain. Let $P^{(a)}_{ij}$  (resp., $P^{(b)}_{ij}$) be the entry $(i,j)$ of matrix $\bP^{(a)}$ (resp., $\bP^{(b)}$).

To evaluate the throughput of a threshold policy, with threshold $T$, we proceed as follows:

\begin{enumerate}
\item we consider a state space  $\Omega$ comprising states of the form $\sigma=(i,j)$
 where $i$ is the number of packets in the buffer,  and
 $j$ is given as follows,
 \begin{equation}
 j=\left\{ 
\begin{array}{ll}
0 & \textrm{no packet being transmitted} \\
a & \textrm{packet being transmitted through channel $a$} \\
b & \textrm{packet being transmitted through channel $b$} \\
\end{array}
\right.
 \end{equation}
 Note that the state space comprises states of the form $(i,j)$ such that,
\begin{equation}
(i,j)=\left\{
\begin{array}{ll}
(0,0) & \textrm{empty system} \\
(i,a) & 0 < i \leq T \\
(i,b) & T < i \leq B-1 \\
\end{array} \right.
\end{equation} 
 \item  Let  $|\Omega|$ be the state space cardinality, $|\Omega|=1+(B-1)=B$, where $B$ is the buffer size. 
  Note that if the buffer size is $B$, we can have at most $B-1$ packets in the system at any decision epoch. This is because
   decision epochs occur immediately after departures.  % Before a departure, the system might have up to $B$ packets, including
%   the one being transmitted.  Immediately after a departure, the system might have up to $B-1$ packets.
\item We build the $|\Omega| \times |\Omega|$ transition matrix $\bP$, corresponding to the given threshold policy.
\begin{enumerate}
\item if $i=0$ and $T=0$, $P_{(i,0),(1,a)}=1$,
\item if $i=0$ and $T\neq0$,  $P_{(i,0),(1,b)}=1$,
\item if $i > 0$ and $0 < j \leq T$, $P_{(i,a),(j,a)}=P^{(a)}_{ij}$ and  $P_{(i,b),(j,a)}=P^{(b)}_{ij}$,
\item if $i > 0$ and $ j > T$, $P_{(i,a),(j,b)}=P^{(a)}_{ij}$ and $P_{(i,b),(j,b)}=P^{(b)}_{ij}$,
\item  if $i > 0$ and $j=0$, $P_{(i,a),(j,0)}=P^{(a)}_{ij}$ and $P_{(i,b),(j,0)}=P^{(b)}_{ij}$.
\end{enumerate}
\item We solve $\bpi = \bpi \bP$ to obtain $\bpi$, the fraction of visits to each state.  The system solution
might be obtained using Gaussian elimination (which involves a matrix inversion), or one of its variants, such as GTH~\cite{grassmann1985regenerative}. 
 Alternatively, the power method can also be employed.
\item We compute the fraction of time at which  the system remains at each state, $\tilde{\bpi}$. 
Let $\tau_{\sigma}$ be the mean time, per visit, at state $\sigma=(i,j)$, 
\begin{equation}
\tau_{(i,j)}=\left\{
\begin{array}{ll}
1/\lambda, & i = 0 \\
1/\mu_a, & j=a \\
1/\mu_b, & j =b
\end{array}
\right.
\end{equation}
Then,\footnote{The rationale is given, for instance, in 
\url{http://www.cs.toronto.edu/~marbach/COURSES/CSC2206_F14/semi_markov.pdf}.}
\begin{eqnarray}
\kappa &=&  \sum_{\forall \sigma \in \Omega} \pi_{\sigma} \tau_\sigma \\
\tilde{\pi}_{\sigma} &=& \pi_\sigma \tau_\sigma/\kappa
\end{eqnarray}

Note that the probability $\sum_{\forall i}\tilde{\pi}_{(i,a)}$  (resp.,  $\sum_{\forall i}\tilde{\pi}_{(i,b)}$) corresponds to the probability that the  system is busy 
transferring a packet through the  channel $a$ (resp., channel $b$).   This follows from the fact that matrix $\bP$ characterizes 
the system dynamics, and its solution yields the system steady state probabilities.   

\item let $\rho_\sigma$ be the instantaneous impulse reward (throughput) obtained after a transition from  state $\sigma=(i,j)$ to any state $\sigma'$,
\begin{equation}
\rho_{(i,j)}=\left\{
\begin{array}{ll}
0, & j = 0 \\
1-p_a, & j=a \\
1-p_b, & j =b
\end{array}
\right.
\end{equation}
\item the throughput $\tilde{\mathcal{T}}$ is given by 
\begin{equation}
\tilde{\mathcal{T}}=\sum_{\forall \sigma \in \Omega} \tilde{\pi}_\sigma \rho_\sigma/\tau_\sigma
\end{equation} 
\end{enumerate}

Consider a long interval of time with duration $I$.  The system remains roughly $\tilde{\pi}_{\sigma} I$ time units at state $\sigma$, 
and visits state $\sigma$ roughly     $\tilde{\pi}_{\sigma} I/\tau_\sigma$ times.  For every visit to state $\sigma$,  $\rho_\sigma$
 is the obtained  impulse reward.  Therefore, the expected impulse reward obtained from visits to state $\sigma$ is
   $\tilde{\pi}_{\sigma}  \rho_\sigma I/\tau_\sigma$. Normalizing by the duration $I$, we obtain the expected impulse reward per time unit
   derived from state $\sigma$, which is  $\tilde{\pi}_{\sigma} \rho_\sigma/\tau_\sigma$. Summing for all states, we obtain the system throughput.

Equivalently,
\begin{equation}
\tilde{\mathcal{T}}=  \left( \sum_{\forall i}\tilde{\pi}_{(i,a)} \right) \mu_a (1-p_a) + \left( \sum_{\forall i}\tilde{\pi}_{(i,b)} \right) \mu_b (1-p_b) 
\end{equation}

Special cases:

\begin{itemize}
\item \textbf{Exponential service times}:  in case the service times are exponentially distributed, policy evaluation consists of computing the expected
instantaneous reward of a continuous-time Markov chain. 
 In this case, tools like Tangram-II~\cite{e2000tangram} can be used to compute the metrics of interest.
 The system can be solved either using the embedded process (as described in this section -- see Figure~\ref{fig:polevalemb}), 
 or directly using the natural process (see Figure~\ref{fig:polevalnat}).
\item \textbf{Deterministic service times}:  special solution methods can be used to efficiently compute the throughput if 
service times are deterministic.  In particular, Tangram-II implements the methods proposed in~\cite{edmundo} for this purpose.
Alternatively, the equations  for the embedded Markov chain follow from the discussion in Section~\ref{sec:det} and can be used
to assess the throughput using the ideas presented in this appendix. 
\item \textbf{Uniform service times}:  the equations for the embedded Markov chain  have been presented in 
Section~\ref{sec:uniform}.  Together
with the algorithm introduced in this appendix, the throughput can be assessed.
\item \textbf{Other service time distributions}: to compute the throughput under other service time distributions, one can rely on tools such as    
 Oris~\cite{ballarini2013transient}.  Table~\ref{tab:orisinput} illustrates % a Petri net used 
the input to Oris. The input comprises three state variables, $W, S_a, S_b$, characterizing the number of waiting packets, and the number of 
packets 
being transmitted through channels $a$ and $b$, respectively.  Note that $W \in \{ 0, 1, \ldots, B-1 \}$,  $S_a \in \{0,1\}$ and $S_b \in \{0,1\}$.
\end{itemize}

\begin{figure}
\center
\includegraphics[scale=0.3]{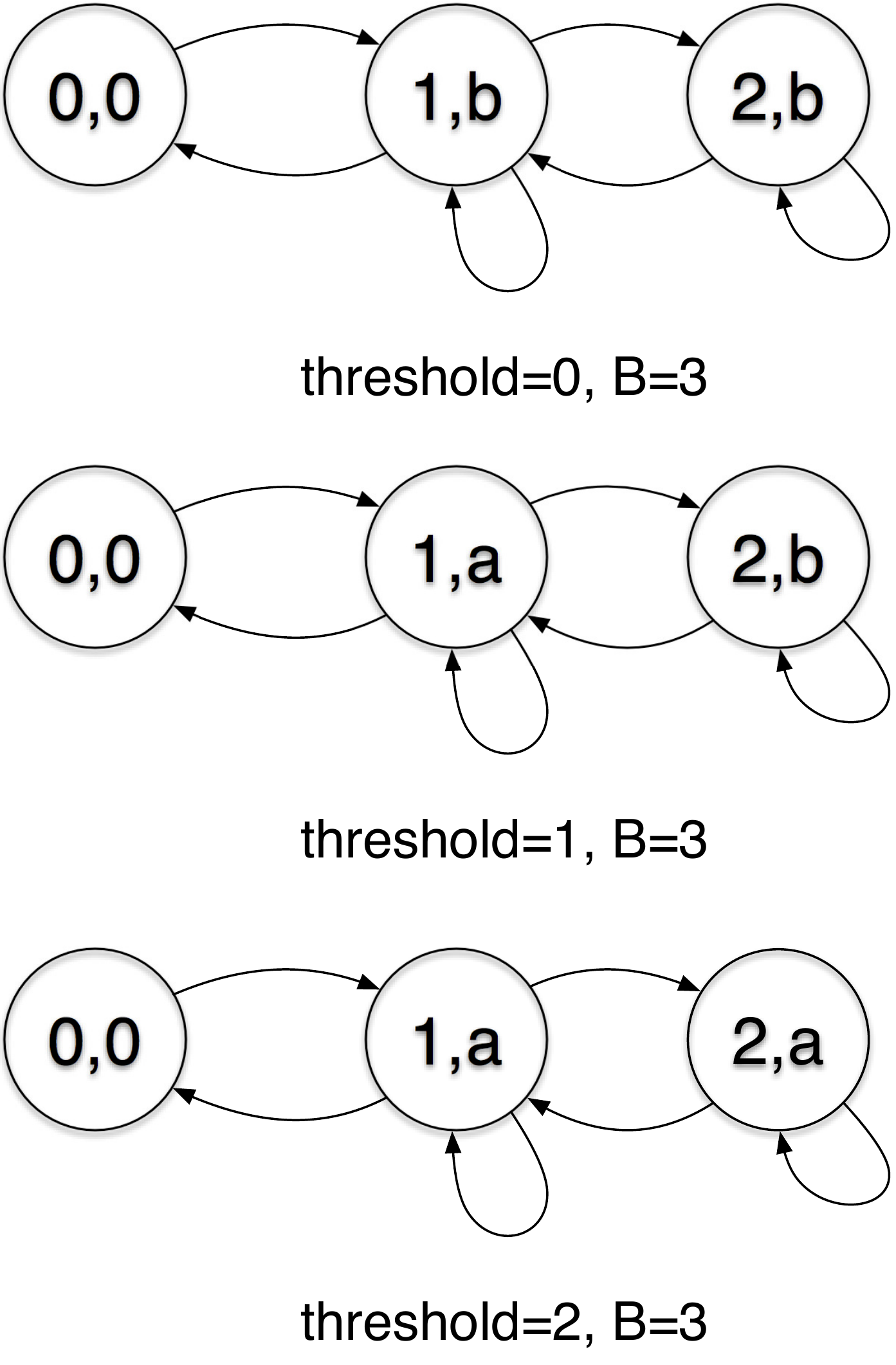}
\caption{Policy evaluation using the embedded  process.  The discrete time Markov chains of the embedded process, as illustrated in this figure,  are applicable for the assessment of the throughput given  generally distributed service times.  }\label{fig:polevalemb}
\end{figure}

 The discrete time Markov chains illustrated in Figure~\ref{fig:polevalemb} correspond to the process embedded at instants of (i) service completions and (ii) arrivals to an empty system, as described in Section~\ref{sec:mdpmodel}.  Note that these Markov  chains model generally distributed service times.  This is in contrast to the discrete time Markov chains subsumed by the   MDP  presented in Figure~\ref{fig:mdpmodel}, which is 
  applicable to exponentially distributed service times.   In the MDP model presented in Figure~\ref{fig:mdpmodel} 
 the embedded points consist of all departures and arrivals.      
 
In the particular case of exponentially distributed service times, the throughput can also be assessed directly using the  natural continuous time process, which is illustrated in Figure~\ref{fig:polevalnat}.  In this section we focused on the use of the embedded process.  A further discussion of the different approaches is presented in~\cite{puterman1994markov}.  

\begin{figure}
\center
\includegraphics[scale=0.3]{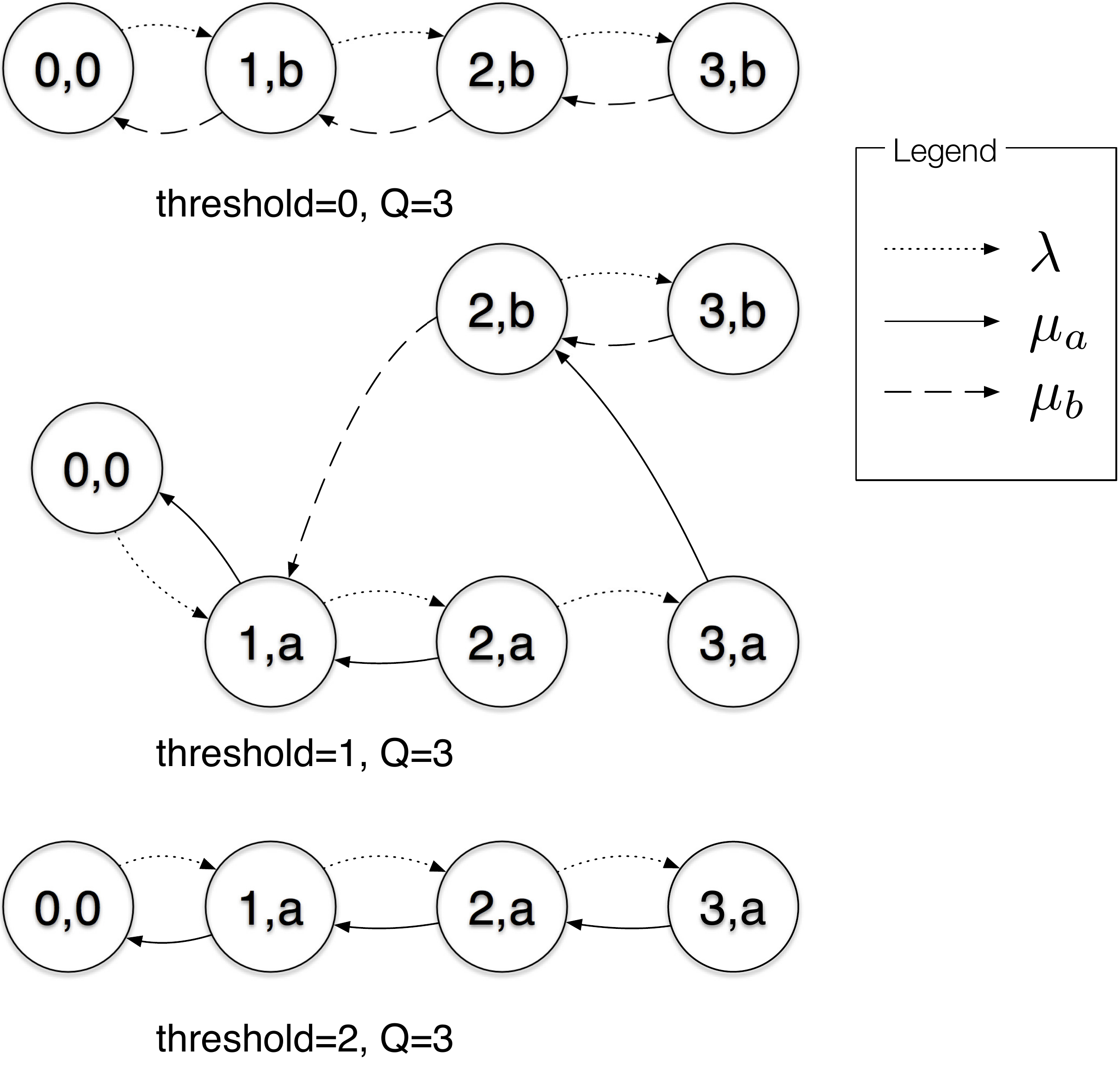}
\caption{Policy evaluation using the natural  process.  The continuous time Markov chains of the natural process, as illustrated in this figure,  are  applicable  for the assessment of the throughput  when service times are exponentially distributed.} \label{fig:polevalnat}
\end{figure}

\begin{comment}
\begin{figure}
\center
\includegraphics[width=\columnwidth]{./petri-crop.pdf}
\caption{Petri net of threshold policy} \label{petri}
\end{figure}
\end{comment}

\begin{table}[h!]
\footnotesize
\center
\caption{Threshold policy characterization}
\begin{tabular}{l|l|l|l}
\hline
event  & inter-event time   &  condition & update \\
\hline
\hline
 packet & exponential,  & $W + {S_a} + S_b < B$ & $W \leftarrow W + 1$ \\
arrival & rate $\lambda$ & & \\
 \hline
 \multicolumn{4}{c}{service start} \\
 \hline
 service $a$ &  0 (instantaneous) & $W > 0$ and & $W  \leftarrow W - 1$, \\
 & &   $S_a=0$ and  & $S_a \leftarrow S_a+1$   \\
 & &  $S_b=0$ and $W \leq T$ & \\
 service $b$ &  0 (instantaneous) & $W > 0$ and & $W  \leftarrow W - 1$,  \\
 & &   $S_a=0$ and   & $S_b \leftarrow S_b+1$ \\
 && $S_b=0$ and $W > T$ & \\
  \hline
  \multicolumn{4}{c}{service completion} \\
  \hline
 service $a$  & general, rate $\mu_a$ & $S_a =1$ & $S_a \leftarrow S_a -1$ \\
service $b$  & general, rate $\mu_b$ & $S_b =1$ & $S_b \leftarrow S_b -1$  \\
\hline 
\end{tabular} \label{tab:orisinput}
\end{table}

% \begin{figure}[h!]
% \includegraphics[width=\columnwidth]{./numericalresults/throughput10.pdf}
% \caption{Throughput for a buffer of size 10}
% \end{figure}

% \begin{figure}[h!]
% \includegraphics[width=\columnwidth]{./numericalresults/throughput50.pdf}
% \caption{Throughput for a buffer of size 50}
% \end{figure}

% \clearpage
% \pagebreak
\fi

\bibliographystyle{IEEEtran} \bibliography{mdprelay}

% Generated by IEEEtran.bst, version: 1.14 (2015/08/26)
\begin{thebibliography}{10}
\providecommand{\url}[1]{#1}
\csname url@samestyle\endcsname
\providecommand{\newblock}{\relax}
\providecommand{\bibinfo}[2]{#2}
\providecommand{\BIBentrySTDinterwordspacing}{\spaceskip=0pt\relax}
\providecommand{\BIBentryALTinterwordstretchfactor}{4}
\providecommand{\BIBentryALTinterwordspacing}{\spaceskip=\fontdimen2\font plus
\BIBentryALTinterwordstretchfactor\fontdimen3\font minus
  \fontdimen4\font\relax}
\providecommand{\BIBforeignlanguage}[2]{{%
\expandafter\ifx\csname l@#1\endcsname\relax
\typeout{** WARNING: IEEEtran.bst: No hyphenation pattern has been}%
\typeout{** loaded for the language `#1'. Using the pattern for}%
\typeout{** the default language instead.}%
\else
\language=\csname l@#1\endcsname
\fi
#2}}
\providecommand{\BIBdecl}{\relax}
\BIBdecl

\bibitem{tse_vishwanath}
D.~Tse and P.~Viswanath, \emph{Fundamentals of wireless communication}.\hskip
  1em plus 0.5em minus 0.4em\relax Cambridge university press, 2005.

\bibitem{zheng_dym_tradeoff}
L.~Zheng and D.~N. Tse, ``Diversity and multiplexing: a fundamental tradeoff in
  multiple-antenna channels,'' \emph{Trans. Info. Theory}, 2003.

\bibitem{karmakar2011diversity}
S.~Karmakar and M.~K. Varanasi, ``The diversity-multiplexing tradeoff of the
  dynamic decode-and-forward protocol on a mimo half-duplex relay channel,''
  \emph{Trans. Info. Theory}, vol.~57, no.~10, 2011.

\bibitem{laneman2004cooperative}
J.~N. Laneman, D.~N. Tse, and G.~W. Wornell, ``Cooperative diversity in
  wireless networks: Efficient protocols and outage behavior,'' \emph{Info.
  Theory, IEEE Trans. on}, vol.~50, no.~12, pp. 3062--3080, 2004.

\bibitem{heath2002diversity}
R.~Heath and A.~Paulraj, ``Diversity versus multiplexing in narrowband mimo
  channels,'' \emph{submitted Dec}, 2002.

\bibitem{tse2004diversity}
D.~N.~C. Tse, P.~Viswanath, and L.~Zheng, ``Diversity-multiplexing tradeoff in
  multiple-access channels,'' \emph{Trans. Info. Theory}, 2004.

\bibitem{liang2011cognitive}
Y.-C. Liang, K.-C. Chen, G.~Y. Li, and P.~Mahonen, ``Cognitive radio networking
  and communications,'' \emph{Trans. Vehic. Techn.}, 2011.

\bibitem{gallager}
R.~Berry and R.~Gallager, ``Communication over fading channels with delay
  constraints,'' \emph{Trans. Info. Theory}, 2002.

\bibitem{cruz}
B.~E. Collins and R.~L. Cruz, ``Transmission policies for time varying channels
  with average delay constraints,'' in \emph{Allerton}, 1999.

\bibitem{goeckel}
A.~Scaglione, D.~L. Goeckel, and J.~N. Laneman, ``Cooperative communications in
  mobile ad hoc networks,'' \emph{Signal Proc. Mag.}, 2006.

\bibitem{tassiulas1992stability}
L.~Tassiulas and A.~Ephremides, ``Stability properties of constrained queueing
  systems and scheduling policies for maximum throughput in multihop radio
  networks,'' \emph{Automatic Control, Trans.}, vol.~37, no.~12, pp.
  1936--1948, 1992.

\bibitem{yeh2007throughput}
E.~M. Yeh and R.~A. Berry, ``Throughput optimal control of cooperative relay
  networks,'' \emph{Info. Theory, Trans.}, 2007.

\bibitem{supittayapornpong2015achieving}
S.~Supittayapornpong and M.~J. Neely, ``Achieving utility-delay-reliability
  tradeoff in stochastic network optimization with finite buffers,''
  \emph{arXiv preprint arXiv:1501.03457}, 2015.

\bibitem{urgaonkar2014delay}
R.~Urgaonkar and M.~J. Neely, ``Delay-limited cooperative communication
  w/reliability constraints in wireless networks,'' \emph{Trans. Info. Theory},
  2014.

\bibitem{altman2007constrained}
E.~Altman, K.~Avratchenkov, N.~Bonneau, M.~Debbah, R.~El-Azouzi, and D.~S.
  Menasch{\'e}, ``Constrained stochastic games in wireless networks,'' in
  \emph{Globecom}.\hskip 1em plus 0.5em minus 0.4em\relax IEEE, 2007, pp.
  315--320.

\bibitem{puterman1994markov}
M.~L. Puterman, \emph{Markov Decision Processes: Discrete Stochastic Dynamic
  Programming}.\hskip 1em plus 0.5em minus 0.4em\relax John Wiley \& Sons,
  Inc., 1994.

\bibitem{koole1998structural}
G.~Koole, ``Structural results for the control of queueing systems using
  event-based dynamic programming,'' \emph{Queueing Systems}, 1998.

\bibitem{hajek1984optimal}
B.~Hajek, ``Optimal control of two interacting service stations,'' \emph{IEEE
  Transactions on Automatic Control}, vol.~29, no.~6, pp. 491--499, 1984.

\bibitem{sennott2009stochastic}
L.~I. Sennott, \emph{Stochastic dynamic programming and the control of queueing
  systems}.\hskip 1em plus 0.5em minus 0.4em\relax John Wiley \& Sons, 2009,
  vol. 504.

\bibitem{walrand1988introduction}
J.~Walrand, \emph{An introduction to queueing networks}.\hskip 1em plus 0.5em
  minus 0.4em\relax Prentice Hall, 1988.

\bibitem{s3paper}
D.~S. Menasch{\'e}, D.~Goeckel, and D.~Towsley, ``Optimal forwarding strategies
  for the relay channel under average delay constraints,'' in \emph{ACM
  Wireless S3}, 2010, pp. 1--4.

\bibitem{bertsekas1995dynamic}
D.~Bertsekas, \emph{Dynamic programming and optimal control}.\hskip 1em plus
  0.5em minus 0.4em\relax Athena Scientific Belmont, MA, 1995, vol.~2, no.~2.

\bibitem{reed1980methods}
M.~Reed and B.~Simon, \emph{Methods of modern mathematical physics: Functional
  analysis}.\hskip 1em plus 0.5em minus 0.4em\relax Gulf Professional
  Publishing, 1980, vol.~1.

\bibitem{edmundo}
E.~de~Souza~e Silva, H.~R. Gail, and R.~R. Muntz, \emph{Efficient solutions for
  a class of non-Markovian models}.\hskip 1em plus 0.5em minus 0.4em\relax
  Springer, 1995.

\bibitem{ballarini2013transient}
P.~Ballarini, N.~Bertrand, A.~Horv{\'a}th, M.~Paolieri, and E.~Vicario,
  ``Transient analysis of networks of stochastic timed automata using
  stochastic state classes,'' in \emph{Quantitative Evaluation of
  Systems}.\hskip 1em plus 0.5em minus 0.4em\relax Springer, 2013, pp.
  355--371.

\bibitem{e2000tangram}
E.~de~Souza~e Silva and R.~M. Leao, ``The {T}angram-{I}{I} environment,'' in
  \emph{Computer Performance Evaluation. Modelling Techniques and Tools}.\hskip
  1em plus 0.5em minus 0.4em\relax Springer, 2000, pp. 366--369.

\bibitem{grassmann1985regenerative}
W.~K. Grassmann, M.~I. Taksar, and D.~P. Heyman, ``Regenerative analysis and
  steady state distributions for markov chains,'' \emph{Operations Research},
  vol.~33, no.~5, pp. 1107--1116, 1985.

\end{thebibliography}

\end{document}